\documentclass[review]{elsarticle}

\usepackage{lineno,hyperref}
\usepackage{amssymb}
\usepackage{color}
\usepackage{amsmath}
\usepackage{tikz}
\usepackage{amsfonts}
\usepackage{mathrsfs}
\usepackage{amsthm}

\usepackage{bbding}
\usepackage{mathrsfs}
\usepackage{amsmath, amsfonts, amssymb, mathrsfs, txfonts}
\usepackage{graphicx, subfigure}
\usepackage{wasysym}
\usepackage{color}
\usepackage{bm}
\usepackage{enumerate}

\usepackage{pifont}
\usepackage{txfonts}
\usepackage{amsmath}
\usepackage{graphicx}
\usepackage{amsfonts}
\usepackage{amssymb}
\usepackage{mathrsfs,psfrag,eepic,epsfig}
\usepackage{epstopdf}

\biboptions{numbers,sort&compress}
\newtheorem{thm}{Theorem}[section]

\newtheorem{prop}[thm]{Proposition}
\theoremstyle{definition}

\modulolinenumbers[5]

\journal{Journal of \LaTeX\ Templates}

\makeatletter \@addtoreset{equation}{section}

\begin{document}

\begin{frontmatter}
\title{Riemann-Hilbert approach and $N$-soliton solutions for a new four-component nonlinear Schr\"{o}dinger equation}
\tnotetext[mytitlenote]{Project supported by the Fundamental Research Fund for the Central Universities under the grant No. 2019ZDPY07.\\
\hspace*{3ex}$^{*}$Corresponding author.\\
\hspace*{3ex}\emph{E-mail addresses}: xmzhou@cumt.edu.cn (X.M. Zhou), sftian@cumt.edu.cn and
shoufu2006@126.com (S. F. Tian),
jinjieyang@cumt.edu.cn (J. J. Yang),
jj-mao@cumt.edu.cn (J. J. Mao)}

\author{Xin-Mei Zhou, Shou-Fu Tian$^{*}$, Jin-Jie Yang and Jin-Jin Mao}
\address{
School of Mathematics and Institute of Mathematical Physics, China University of Mining and Technology,\\ Xuzhou 221116, People's Republic of China\\
}

\begin{abstract}

A new four-component nonlinear Schr\"{o}dinger equation is first proposed in this work and studied by Riemann-Hilbert approach. Firstly, we derive a Lax pair associated with a $5\times5$ matrix spectral problem for the four-component nonlinear Schr\"{o}dinger equation. Then based on the Lax pair, we analyze the spectral problem and the analytical properties of the Jost functions, from which the Riemann-Hilbert problem of the equation is successfully established. Moreover, we obtain the $N$-soliton solutions of the   equation  by solving the Riemann-Hilbert problem without reflection. Finally, we derive two special cases of the solutions to the equation for $N=1$ and $N=2$, and the local structure and dynamic behavior of the one-and two-soliton solutions are analyzed graphically.

\end{abstract}


\end{frontmatter}

%\linenumbers

\section{Introduction}
The nonlinear Schr\"{o}dinger equation (NLS) is an important integrable model. It is closely related to many nonlinear problems in theoretical physics such as nonlinear optics and ion acoustic waves of plasmas. Some higher-order coupled NLS equations are proposed, to describe more deep physical effects, including self-deepening, self-frequency shifting, and cubic-quintic nonlinearity. Among the different solutions of these models, soliton solutions play a crucial role in explaining some related complex nonlinear phenomena. With the development of nonlinear science, there are many ways to find solutions for nonlinear integrable models, including inverse scattering transform \cite{1}, Darboux transform \cite{2}, Hirota bilinear method \cite{3}, Lie group method \cite{4}, etc. Among them, inverse scattering transform method is one of the most effective tools for solving the initial value problem of nonlinear integrable systems to get the soliton solutions.
For second-order spectral problems, inverse scattering theory is equivalent to Riemann-Hilbert (RH) approach, but for higher-order spectral problems the development of inverse scattering theory  is not perfect, part of the inverse scattering problem needs to be transformed into RH problem. RH approach is developed by Zakharov et al \cite{5}, applied to integrable systems \cite{6}-\cite{31} as a more general method than inverse scattering method.  This method has been successfully used to study the integrable system with single component. However, to the best of authors' knowledge, there are very few studies on the multi-component problems. The well-known general two-component coupled nonlinear Schr\"{o}dinger equation of the form \cite{32}
\begin{equation}\label{NLS-0}
\left\{
\begin{aligned}
&ip_{t}+p_{xx}+2(a|p|^{2}+c|q|^{2}+bpq^{*}+b^{*}qp^{*})p=0,\\
&iq_{t}+q_{xx}+2(a|p|^{2}+c|q|^{2}+bpq^{*}+b^{*}qp^{*})q=0,
\end{aligned}
\right.
\end{equation}
where $a$ and $c$ are real constants, $b$ is a complex constant, and $``\ast"$ denotes complex conjugation.
In physics, $a$ and $c$ describe the SPM and XPM effects, and $b$ and $b^{*}$ describe the four-wave mixing effects.

In this work, we first propose an interesting equation named by 
 a new four-component nonlinear Schr\"{o}dinger (FCNLS) equations
\begin{equation}\label{NLS-1}
\left\{
\begin{aligned}
&iq_{1t}+q_{1xx}-2[a_{11}|q_{1}|^{2}+a_{22}|q_{2}|^{2}+a_{33}|q_{3}|^{2}+a_{44}|q_{4}|^{2}\\
&+2\mbox{Re}(a_{12}q_{1}^{*}q_{2}+a_{13}q_{1}^{*}q_{3}+a_{14}q_{1}^{*}q_{4}+a_{23}q_{2}^{*}q_{3}+a_{24}q_{2}^{*}q_{4}+a_{34}q_{3}^{*}q_{4})]q_{1}=0,\\
&iq_{2t}+q_{2xx}-2[a_{11}|q_{1}|^{2}+a_{22}|q_{2}|^{2}+a_{33}|q_{3}|^{2}+a_{44}|q_{4}|^{2}\\
&+2\mbox{Re}(a_{12}q_{1}^{*}q_{2}+a_{13}q_{1}^{*}q_{3}+a_{14}q_{1}^{*}q_{4}+a_{23}q_{2}^{*}q_{3}+a_{24}q_{2}^{*}q_{4}+a_{34}q_{3}^{*}q_{4})]q_{2}=0,\\
&iq_{3t}+q_{3xx}-2[a_{11}|q_{1}|^{2}+a_{22}|q_{2}|^{2}+a_{33}|q_{3}|^{2}+a_{44}|q_{4}|^{2}\\
&+2\mbox{Re}(a_{12}q_{1}^{*}q_{2}+a_{13}q_{1}^{*}q_{3}+a_{14}q_{1}^{*}q_{4}+a_{23}q_{2}^{*}q_{3}+a_{24}q_{2}^{*}q_{4}+a_{34}q_{3}^{*}q_{4})]q_{3}=0,\\
&iq_{4t}+q_{4xx}-2[a_{11}|q_{1}|^{2}+a_{22}|q_{2}|^{2}+a_{33}|q_{3}|^{2}+a_{44}|q_{4}|^{2}\\
&+2\mbox{Re}(a_{12}q_{1}^{*}q_{2}+a_{13}q_{1}^{*}q_{3}+a_{14}q_{1}^{*}q_{4}+a_{23}q_{2}^{*}q_{3}+a_{24}q_{2}^{*}q_{4}+a_{34}q_{3}^{*}q_{4})]q_{4}=0,
\end{aligned}
\right.
\end{equation}
where $a_{11}$, $a_{22}$, $a_{33}$ and $a_{44}$ are real constants, $a_{12}$, $a_{13}$, $a_{14}$, $a_{23}$, $a_{24}$ and $a_{34}$ are complex constant, $``*"$ denotes complex conjugation, and $``\mbox{Re}"$ denotes the real part.
The FCNLS equation includes group velocity dispersion, self-phase modulation, cross-phase modulation and paired tunnel modulation. This equation can be reduced to the three-component nonlinear Schr\"{o}dinger equation \eqref{NLS-3} given by 
\begin{equation}\label{NLS-3}
\left\{
\begin{aligned}
\begin{split}
&iq_{1t}+q_{1xx}-2[a|q_{1}|^{2}+c|q_{2}|^{2}+f|q_{3}|^{2}+2\mbox{Re}(bq_{1}^{*}q_{2}+dq_{1}^{*}q_{3}+eq_{2}^{*}q_{3})]q_{1}=0,\\
&iq_{2t}+q_{2xx}-2[a|q_{1}|^{2}+c|q_{2}|^{2}+f|q_{3}|^{2}+2\mbox{Re}(bq_{1}^{*}q_{2}+dq_{1}^{*}q_{3}+eq_{2}^{*}q_{3})]q_{2}=0,\\
&iq_{3t}+q_{3xx}-2[a|q_{1}|^{2}+c|q_{2}|^{2}+f|q_{3}|^{2}+2\mbox{Re}(bq_{1}^{*}q_{2}+dq_{1}^{*}q_{3}+eq_{2}^{*}q_{3})]q_{3}=0,
\end{split}
\end{aligned}
\right.
\end{equation}
where $a$, $c$ and $f$ are real constants, $b$, $d$ and $e$ are complex constant,  $``*"$ denotes complex conjugation, and $``\mbox{Re}"$ denotes the real part. This equation can be reduced to equation \eqref{NLS-0}. 
Eq. \eqref{NLS-0} is studied by extending the Fokas unified approach by Yan in \cite{33}.
Eq. \eqref{NLS-3} can be reduced to different three-component NLS equations with the different conditions of the six parameters $a$, $b$, $c$, $d$, $e$ and $f$, such as
\begin{itemize}
\item The three-component focused NLS equation for $a=c=f=-1$ and $b=d=e=0$.

\item The three-component defocused NLS equation for $a=c=f=1$ and $b=d=e=0$.

\item The three-component mixed NLS equation for $a=-1$, $c=f=1$ and $b=d=e=0$ or $a=1$, $c=f=-1$ and $b=d=e=0$.
If one takes other parameter values, equation \eqref{NLS-3} can reduced to other three-component NLS equations.
\end{itemize}

The main purpose of this work is to study the RH problem  for the FCNLS equation \eqref{NLS-1} by first deriving its Lax pair, and obtain it's $N$-soliton solutions. The main results of the present work are as follows.
\begin{thm}
The FCNLS equation \eqref{NLS-1} has the following $N$-soliton solutions
\begin{equation*}
q_{1}=-2i\frac{\det F}{\det M},~~~q_{2}=-2i\frac{\det G}{\det M},~~~q_{3}=-2i\frac{\det H}{\det M},~~~q_{4}=-2i\frac{\det K}{\det M},
\end{equation*}
where $M$, $F$, $G$, $H$ and $K$ can get from \eqref{lax-64}, \eqref{lax-74}, \eqref{lax-75}, \eqref{lax-76} and \eqref{lax-761}, respectively.
\end{thm}
\begin{thm}
The FCNLS equation \eqref{NLS-1} has the following one-soliton solutions
\begin{align*}
&q_{1}=2m_{1}\alpha_{1}\gamma_{1}^{\ast}e^{-\xi_{1}}e^{-2in_{1}x-4in_{1}^{2}t}\mbox{sech}(2m_{1}x+4im_{1}^{2}t+8n_{1}m_{1}t+\xi_{1}),\\
&q_{2}=2m_{1}\beta_{1}\gamma_{1}^{\ast}e^{-\xi_{1}}e^{-2in_{1}x-4in_{1}^{2}t}\mbox{sech}(2m_{1}x+4im_{1}^{2}t+8n_{1}m_{1}t+\xi_{1}),\\
&q_{3}=2m_{1}\tau_{1}\gamma_{1}^{\ast}e^{-\xi_{1}}e^{-2in_{1}x-4in_{1}^{2}t}\mbox{sech}(2m_{1}x+4im_{1}^{2}t+8n_{1}m_{1}t+\xi_{1}),\\
&q_{4}=2m_{1}\zeta_{1}\gamma_{1}^{\ast}e^{-\xi_{1}}e^{-2in_{1}x-4in_{1}^{2}t}\mbox{sech}(2m_{1}x+4im_{1}^{2}t+8n_{1}m_{1}t+\xi_{1}),
\end{align*}
where $n_{1}$ and $m_{1}$ are arbitrary real numbers, $\alpha_{1}$, $\beta_{1}$, $\tau_{1}$, $\zeta_{1}$ and $\gamma_{1}$ are arbitrary imaginary numbers, and $e^{-\xi_{1}}$ can be obtained from \eqref{111}.
\end{thm}
\begin{thm}
The FCNLS equation \eqref{NLS-1} has the following two-soliton solutions
\begin{equation*}
\begin{split}
q_{1}=&\frac{-2i}{M_{11}M_{22}-M_{12}M_{21}}(-\alpha_{1}\gamma_{1}^{\ast}e^{\theta_{1}-\theta_{1}^{\ast}}M_{22}+\alpha_{1}\gamma_{2}^{\ast}e^{\theta_{1}-\theta_{2}^{\ast}}M_{12}\\
&+\alpha_{2}\gamma_{1}^{\ast}e^{\theta_{2}-\theta_{1}^{\ast}}M_{21}-\alpha_{2}\gamma_{2}^{\ast}e^{\theta_{2}-\theta_{2}^{\ast}}M_{11}),\\
q_{2}=&\frac{-2i}{M_{11}M_{22}-M_{12}M_{21}}(-\beta_{1}\gamma_{1}^{\ast}e^{\theta_{1}-\theta_{1}^{\ast}}M_{22}+\beta_{1}\gamma_{2}^{\ast}e^{\theta_{1}-\theta_{2}^{\ast}}M_{12}\\
&+\beta_{2}\gamma_{1}^{\ast}e^{\theta_{2}-\theta_{1}^{\ast}}M_{21}-\beta_{2}\gamma_{2}^{\ast}e^{\theta_{2}-\theta_{2}^{\ast}}M_{11}),\\
q_{3}=&\frac{-2i}{M_{11}M_{22}-M_{12}M_{21}}(-\tau_{1}\gamma_{1}^{\ast}e^{\theta_{1}-\theta_{1}^{\ast}}M_{22}+\tau_{1}\gamma_{2}^{\ast}e^{\theta_{1}-\theta_{2}^{\ast}}M_{12}\\
&+\tau_{2}\gamma_{1}^{\ast}e^{\theta_{2}-\theta_{1}^{\ast}}M_{21}-\tau_{2}\gamma_{2}^{\ast}e^{\theta_{2}-\theta_{2}^{\ast}}M_{11}),\\
q_{4}=&\frac{-2i}{M_{11}M_{22}-M_{12}M_{21}}(-\zeta_{1}\gamma_{1}^{\ast}e^{\theta_{1}-\theta_{1}^{\ast}}M_{22}+\zeta_{1}\gamma_{2}^{\ast}e^{\theta_{1}-\theta_{2}^{\ast}}M_{12}\\
&+\zeta_{2}\gamma_{1}^{\ast}e^{\theta_{2}-\theta_{1}^{\ast}}M_{21}-\zeta_{2}\gamma_{2}^{\ast}e^{\theta_{2}-\theta_{2}^{\ast}}M_{11}),
\end{split}
\end{equation*}
where $\theta_{1}-\theta_{1}^{\ast}=-2in_{1}x-4in_{1}^{2}t$, $\theta_{1}+\theta_{1}^{\ast}=2m_{1}x+4im_{1}^{2}t+8n_{1}m_{1}t$, $\alpha_{2}$, $\beta_{2}$, $\tau_{2}$, $\zeta_{2}$ and $\gamma_{2}$ are arbitrary imaginary numbers, and $M_{11}$, $M_{12}$, $M_{21}$ and $M_{22}$ can be obtained from \eqref{lax-100}.
\end{thm}

The structure of this work is as follows. In the second part, we derive a Lax pair associated with a $5\times5$ matrix spectral problem for the FCNLS equation \eqref{NLS-1}. Then based on the Lax pair with a $5\times5$ matrix, we analyze the spectral problem and the analytical properties of the Jost functions. In the third part, we establish the RH problem based on the previous conclusions. Next, we give the symmetry of the scattering matrix, and study the temporal and spatial evolution of the scattering data. In the fourth part, by solving the RH problem, we obtain the $N$-soliton solutions of the FCNLS equation \eqref{NLS-1}, and analyze the propagation behaviors of one-soliton solutions and two-soliton solutions. Finally,   some conclusions are presented in the last section.

\section{Spectral analysis}

\subsection{The Lax Pair and eigenfunction }

We first derive the Lax pair of the FCNLS equation \eqref{NLS-1} via the following theorem. 
\begin{thm}
The FCNLS equation \eqref{NLS-1} admits the following Lax pair
\begin{equation}\label{lax-1}
\Phi_{x}=U\Phi,~~~~~~\Phi_{t}=V\Phi,
\end{equation}
where $\Phi$ is a column vector function, and matrices $U$ and $V$ are written as
\begin{align}\label{lax-2}
&U=\left(\begin{array}{ccccc}
   -\dot{\imath}\lambda & 0 & 0 & 0 & q_{1} \\
  0 &-\dot{\imath}\lambda & 0 & 0 & q_{2} \\
   0 & 0 & -\dot{\imath}\lambda & 0 & q_{3} \\
    0 & 0 & 0 & -\dot{\imath}\lambda & q_{4} \\
     p_{1} & p_{2} & p_{3} & p_{4} &\dot{\imath}\lambda\\
\end{array}\right),\\
&V=-2\dot{\imath}\lambda^{2}\Lambda+2\lambda P+V_{0}
,~~V_{0}=-\dot{\imath}(P_{x}+P^{2}),
\end{align}
here $\lambda$ being the spectral parameter and $p_{1}=a_{11}q_{1}^{*}+a_{21}q_{2}^{*}+a_{31}q_{3}^{*}+a_{41}q_{4}^{*}$,~$p_{2}=a_{21}^{\ast}q_{1}^{*}+a_{22}q_{2}^{*}+a_{32}q_{3}^{*}+a_{42}q_{4}^{*}$,~$p_{3}=a_{31}^{\ast}q_{1}^{*}+a_{32}^{\ast}q_{2}^{*}+a_{33}q_{3}^{*}+a_{43}q_{4}^{*}$,~$p_{4}=a_{41}^{\ast}q_{1}^{*}+a_{42}^{\ast}q_{2}^{*}+a_{43}^{\ast}q_{3}^{*}+a_{44}q_{4}^{*}$,
with
\begin{equation}\label{lax-3}
\Lambda=\left(\begin{array}{ccccc}
   1 & 0 & 0 & 0 & 0\\
  0 &1 & 0 & 0 & 0\\
   0 & 0 & 1 & 0 & 0\\
   0 & 0 & 0 & 1 & 0\\
    0 & 0 & 0 & 0 &-1\\
    \end{array}\right),~~~P=\left(\begin{array}{ccccc}
   0 & 0 & 0 & 0& q_{1} \\
  0 &0 & 0 & 0 &q_{2}\\
   0 & 0 & 0 & 0& q_{3} \\
    0 & 0 & 0 & 0& q_{4} \\
    p_{1} & p_{2} & p_{3}& p_{4} &0\\
    \end{array}\right).
\end{equation}
\end{thm}
\begin{proof}
 The compatibility condition of the two equations in Eqs.\eqref{lax-1}
 \begin{equation}
 U_{t}-V_{x}+[U,V]=0,
 \end{equation}
 with $[U,V]=UV-VU$, which is reduced to the FCNLS equation \eqref{NLS-1}.
\end{proof}
Then we obtain that
\begin{align}\label{lax-5.1}
\begin{split}
&\Phi_{x}+\dot{\imath}\lambda\Lambda\Phi=P\Phi,\\
&\Phi_{t}+2\dot{\imath}\lambda^{2}\Lambda\Phi=Q\Phi,
\end{split}
\end{align}
where $Q=2\lambda P+V_{0}$.

From Eqs.\eqref{lax-5.1} when $|x|\rightarrow\infty$, one has
\begin{equation}\label{lax-7}
\Phi\propto e^{-\dot{\imath}\lambda\Lambda x-2\dot{\imath}\lambda^{2}\Lambda t}.
\end{equation}
Letting
\begin{equation}\label{lax-8}
\mu=\Phi e^{\dot{\imath}\lambda\Lambda x+2\dot{\imath}\lambda^{2}\Lambda t},
\end{equation}
then we can get the equivalent Lax pair
\begin{align}\label{lax-9.1}
\begin{split}
&\mu_{x}+\dot{\imath}\lambda[\Lambda,\mu]=P\mu,\\
&\mu_{t}+2\dot{\imath}\lambda^{2}[\Lambda,\mu]=Q\mu,
\end{split}
\end{align}
where~$[\Lambda,\mu]=\Lambda\mu-\mu\Lambda$~is the commutator. We can get the following full differential
\begin{equation}\label{lax-10}
d\left(e^{\dot{\imath}(\lambda x+2\lambda^{2}t)\bar{\Lambda}}\mu\right)=e^{\dot{\imath}(\lambda x+2\lambda^{2}t)\bar{\Lambda}}[\left(Pdx+Qdt\right)\mu],
\end{equation}
where~$e^{\lambda\bar{\Lambda}}\mu=e^{\lambda\Lambda}\mu e^{-\lambda\Lambda}$.

\subsection{Asymptotic analysis}
To formulate an RH problem, we seek solutions of the spectral problem with the
$5\times5$ unit matrix as $\lambda\rightarrow\infty$. Let us consider the solution of Eq.\eqref{lax-10} as follows
\begin{equation}\label{lax-11}
\mu=\mu^{(0)}+\frac{\mu^{(1)}}{\lambda}+\frac{\mu^{(2)}}{\lambda^{2}}+o\left(\frac{1}{\lambda^{3}}\right),~~\lambda\rightarrow\infty,
\end{equation}
where $\mu^{(0)}$, $\mu^{(1)}$ and $\mu^{(2)}$ are independent of $\lambda$. Substituting Eq.\eqref{lax-11} into Eqs.\eqref{lax-9.1}, and comparing the same order of frequency for $\lambda$, we obtain
\begin{align}\label{lax-12}
\begin{split}
&o(1)~:~\mu_{x}^{(0)}+i\lambda[\Lambda,\mu^{(1)}]=P\mu^{(0)},\\
&o(\lambda)~:~i\lambda[\Lambda,\mu^{(0)}]=0,\\
&o(\lambda)~:~2i\lambda^{2}[\Lambda,\mu^{(1)}]=2\lambda P\mu^{(0)}.
\end{split}
\end{align}
Eqs.\eqref{lax-12} implies $\mu^{(0)}$ is a diagonal matrix, $\mu_{x}^{(0)}=0$. This means that $\mu^{(0)}$ is not related to $x$. Then
\begin{equation}\label{lax-13}
\begin{split}
I&=\lim_{\lambda\rightarrow\infty}\lim_{|x|\rightarrow\infty}\mu=\mu^{(0)}.
\end{split}
\end{equation}

Now, two solutions $\mu_{\pm}=\mu_{\pm}(x,\lambda)$ are constructed for Eq.\eqref{lax-9.1}
\begin{align}\label{lax-14.1}
\begin{split}
&\mu_{+}=\left([\mu_{+}]_{1},[\mu_{+}]_{2},[\mu_{+}]_{3},[\mu_{+}]_{4},[\mu_{+}]_{5}\right),\\
&\mu_{-}=\left([\mu_{-}]_{1},[\mu_{-}]_{2},[\mu_{-}]_{3},[\mu_{-}]_{4},[\mu_{-}]_{5}\right),
\end{split}
\end{align}
with the asymptotic conditions
\begin{align}\label{lax-15}
\begin{split}
&\mu_{+}\rightarrow I~~as~~x\rightarrow +\infty,\\
&\mu_{-}\rightarrow I~~as~~x\rightarrow -\infty,
\end{split}
\end{align}
here each $[\mu_{+}]_{l} (l=1,2,3,4,5)$ denotes the l-th column of the matrices $[\mu_{\pm}]$,  respectively. The symbol $I$ is the $5\times5$ unit matrix, and the two solutions $[\mu_{\pm}]$ are uniquely determined by the Volterra integral equations for $\lambda\in R$
\begin{align}\label{lax-16}
\begin{split}
&\mu_{+}(x,\lambda)=I-\int^{+\infty}_{x} e^{-i\lambda\Lambda(x-y)}P(y)\mu_{+}(y,\lambda)e^{i\lambda\Lambda(x-y)}dy,\\
&\mu_{-}(x,\lambda)=I+\int^{x}_{-\infty} e^{-i\lambda\Lambda(x-y)}P(y)\mu_{-}(y,\lambda)e^{i\lambda\Lambda(x-y)}dy.
\end{split}
\end{align}
Then we analysis the Eqs.\eqref{lax-16},
\begin{equation}\label{lax-18}
e^{-i\lambda\Lambda(x-y)}Pe^{i\lambda\Lambda(x-y)}=\left(\begin{array}{ccccc}
   0 & 0 & 0 & 0 & q_{1}e^{-2i\lambda(x-y)} \\
  0 &0 & 0 & 0 & q_{2}e^{-2i\lambda(x-y)} \\
   0 & 0 & 0 & 0 & q_{3}e^{-2i\lambda(x-y)} \\
   0 & 0 & 0 & 0 & q_{4}e^{-2i\lambda(x-y)} \\
    p_{1}e^{2i\lambda(x-y)} & p_{2}e^{2i\lambda(x-y)} & p_{3}e^{2i\lambda(x-y)} & p_{4}e^{2i\lambda(x-y)} &0\\
    \end{array}\right).
\end{equation}
To find the analytic area of each column, we just consider $\mbox{Re}[2i\lambda(x-y)]<0$ and $\mbox{Re}[-2i\lambda(x-y)]<0$. One obtains $[\mu_{-}]_{1}$, $[\mu_{-}]_{2}$, $[\mu_{-}]_{3}$, $[\mu_{-}]_{4}$ and $[\mu_{+}]_{5}$ are analytic in the upper half-plane $C^{+}$. Similarly, $[\mu_{+}]_{1}$, $[\mu_{+}]_{2}$, $[\mu_{+}]_{3}$, $[\mu_{+}]_{4}$ and $[\mu_{-}]_{5}$ are analytic in the lower half-plane $C^{-}$.
Now we investigate the properties of $\mu_{\pm}$. Since $\mbox{tr}(P) = 0$ and Liouville's formula, we know that the determinants of $\mu_{\pm}$ are independent of the variable $x$. Therefore we obtain from Eq.\eqref{lax-15} that
\begin{equation}\label{lax-19}
\det\mu_{\pm}=1,~~~~~\lambda\in R.
\end{equation}

Since $\mu_{\pm}E$ are both matrix solutions of the spectral problem Eqs.\eqref{lax-9.1}, where $E=e^{-i\lambda\Lambda x}$. Therefore, these two solutions are interdependent, and they must be related by a scattering matrix $S(\lambda)=(s_{kj})_{5\times5}$
\begin{equation}\label{lax-20}
\mu_{-}E=\mu_{+}ES(\lambda),~~~~~\lambda\in R.
\end{equation}
From Eq.\eqref{lax-19} and Eq.\eqref{lax-20}, we have
\begin{equation}\label{lax-21}
\det S(\lambda)=1,~~~~~\lambda\in R.
\end{equation}

To formulate an RH problem for the FCNLS equations \eqref{NLS-1}, we consider the inverse matrices of $\mu_{\pm}$ as
\begin{equation}\label{lax-22}
\mu^{-1}_{\pm}=\left(\begin{array}{c}
 [\mu_{\pm}^{-1}]^{1}\\

  [\mu_{\pm}^{-1}]^{2}\\

   [\mu_{\pm}^{-1}]^{3}\\

    [\mu_{\pm}^{-1}]^{4}\\

     [\mu_{\pm}^{-1}]^{5}\\
    \end{array}\right),
\end{equation}
where each $[\mu_{\pm}^{-1}]^{l}, (l=1,2,3,4,5)$ denotes the l-th row of $\mu_{\pm}^{-1}$, respectively.
\begin{thm}
Letting
\begin{equation*}
S(\lambda)=\left(\begin{array}{ccccc}
s_{11}&s_{12}&s_{13}&s_{14}&s_{15}\\
s_{21}&s_{22}&s_{23}&s_{24}&s_{25}\\
s_{31}&s_{32}&s_{33}&s_{34}&s_{35}\\
s_{41}&s_{42}&s_{43}&s_{44}&s_{45}\\
s_{51}&s_{52}&s_{53}&s_{54}&s_{55}\\
    \end{array}\right),
\end{equation*}
\begin{equation*}
R(\lambda)=\left(\begin{array}{ccccc}
r_{11}&r_{12}&r_{13}&r_{14}&r_{15}\\
r_{21}&r_{22}&r_{23}&r_{24}&r_{25}\\
r_{31}&r_{32}&r_{33}&r_{34}&r_{35}\\
r_{41}&r_{42}&r_{43}&r_{44}&r_{45}\\
r_{51}&r_{52}&r_{53}&r_{54}&r_{55}\\
    \end{array}\right),
\end{equation*}
$s_{11}$, $s_{12}$, $s_{13}$, $s_{14}$, $s_{21}$, $s_{22}$, $s_{23}$, $s_{24}$, $s_{31}$, $s_{32}$, $s_{33}$, $s_{34}$, $s_{41}$, $s_{42}$, $s_{43}$ and $s_{44}$ are analytic for $\lambda\in C^{+}$, $s_{55}$ is analytic in $C^{-}$. $r_{55}$ is analytic in $C^{+}$, $r_{11}$, $r_{12}$, $r_{13}$, $r_{14}$, $r_{21}$, $r_{22}$, $r_{23}$, $r_{24}$, $r_{31}$, $r_{32}$, $r_{33}$, $r_{34}$, $r_{41}$, $r_{42}$, $r_{43}$ and $r_{44}$ are analytic for $\lambda\in C^{-}$.
\end{thm}
\begin{proof}
Using Eq.\eqref{lax-9.1} it is easy to verify that $\mu_{\pm}^{-1}$ satisfy the equation of K
\begin{equation}\label{lax-23}
K_{x}=-i\lambda[\Lambda,K]-KP.
\end{equation}
 According to \eqref{lax-20}, it's easy to find
\begin{equation}\label{lax-24}
E^{-1}\mu_{-}^{-1}=R(\lambda)E^{-1}\mu_{+}^{-1}.
\end{equation}
From Eq.\eqref{lax-20}, we have $E^{-1}\mu_{+}^{-1}\mu_{-}E=S(\lambda)$.
\begin{equation*}\label{lax-25}
S(\lambda)=E^{-1}\left(\begin{array}{ccccc}
 [\mu_{+}^{-1}]^{1}[\mu_{-}]_{1}&[\mu_{+}^{-1}]^{1}[\mu_{-}]_{2}&[\mu_{+}^{-1}]^{1}[\mu_{-}]_{3}&[\mu_{+}^{-1}]^{1}[\mu_{-}]_{4}&[\mu_{+}^{-1}]^{1}[\mu_{-}]_{5}\\

  [\mu_{+}^{-1}]^{2}[\mu_{-}]_{1}&[\mu_{+}^{-1}]^{2}[\mu_{-}]_{2}&[\mu_{+}^{-1}]^{2}[\mu_{-}]_{3}&[\mu_{+}^{-1}]^{2}[\mu_{-}]_{4}&[\mu_{+}^{-1}]^{2}[\mu_{-}]_{5}\\

   [\mu_{+}^{-1}]^{3}[\mu_{-}]_{1}&[\mu_{+}^{-1}]^{3}[\mu_{-}]_{2}&[\mu_{+}^{-1}]^{3}[\mu_{-}]_{3}&[\mu_{+}^{-1}]^{3}[\mu_{-}]_{4}&[\mu_{+}^{-1}]^{3}[\mu_{-}]_{5}\\

    [\mu_{+}^{-1}]^{4}[\mu_{-}]_{1}&[\mu_{+}^{-1}]^{4}[\mu_{-}]_{2}&[\mu_{+}^{-1}]^{4}[\mu_{-}]_{3}&[\mu_{+}^{-1}]^{4}[\mu_{-}]_{4}&[\mu_{+}^{-1}]^{4}[\mu_{-}]_{5}\\

     [\mu_{+}^{-1}]^{5}[\mu_{-}]_{1}&[\mu_{+}^{-1}]^{5}[\mu_{-}]_{2}&[\mu_{+}^{-1}]^{5}[\mu_{-}]_{3}&[\mu_{+}^{-1}]^{5}[\mu_{-}]_{4}&[\mu_{+}^{-1}]^{5}[\mu_{-}]_{5}\\
    \end{array}\right)E,
\end{equation*}
according to the analytic property of $\mu_{+}^{-1}$ and $\mu_{-}$, we can proof the theorem. The matrix $R(\lambda)$ can be analyzed in the same way.
\end{proof}

\section{ Riemann-Hilbert problem}

In this part, an RH problem is formulated by using the properties of $\mu_{\pm}$. We construct matrix function $P_{1}=P_{1}(x,\lambda)$ and $P_{2}=P_{2}(x,\lambda)$. The function $P_{1}=P_{1}(x,\lambda)$ is analytic in $C^{+}$, and the function $P_{2}=P_{2}(x,\lambda)$ is analytic in $C^{-}$. 

Let
\begin{equation}\label{lax-26}
P_{1}=\left([\mu_{-}]_{1},[\mu_{-}]_{2},[\mu_{-}]_{3},[\mu_{-}]_{4},[\mu_{+}]_{5}\right),
\end{equation}
\begin{equation}\label{lax-27}
P_{2}=\left(\begin{array}{c}
 [\mu_{-}^{-1}]^{1}\\
 
  [\mu_{-}^{-1}]^{2}\\
  
   [\mu_{-}^{-1}]^{3}\\
   
    [\mu_{-}^{-1}]^{4}\\
    
     [\mu_{+}^{-1}]^{5}\\
    \end{array}\right),
\end{equation}
with
\begin{align}\label{lax-28}
\begin{split}
&P_{1}\rightarrow I,~~~as~~\lambda\rightarrow +\infty,\\
&P_{2}\rightarrow I,~~~as~~\lambda\rightarrow -\infty.
\end{split}
\end{align}

At present, we restrict $P_{1}$ to the left-hand side of the real $\lambda$-axis as $P_{+}$, and the restrict $P_{2}$ to the right-hand side of the real $\lambda$-axis as $P_{-}$. On the real line, they are meet
 \begin{equation}\label{lax-29}
P_{-}(x,\lambda)P_{+}(x,\lambda)=G(x,\lambda),~~~\lambda\in R,
\end{equation}
with
\begin{equation}\label{lax-30}
G(x,\lambda)=\left(\begin{array}{ccccc}
   1 & 0 & 0 & 0 & r_{15}e^{-2i\lambda x} \\
  0 &1 & 0 & 0 & r_{25}e^{-2i\lambda x} \\
   0 & 0 & 1 & 0 & r_{35}e^{-2i\lambda x} \\
   0 & 0  & 0 & 1 & r_{45}e^{-2i\lambda x} \\
    s_{51}e^{2i\lambda x} & s_{52}e^{2i\lambda x} & s_{53}e^{2i\lambda x} & s_{54}e^{2i\lambda x} &1\\
    \end{array}\right).
\end{equation}
According to Eq.\eqref{lax-28}, we obtain the canonical normalization conditions as follows
\begin{align*}
&P_{1}\rightarrow I,~~~as~~\lambda\rightarrow +\infty,\\
&P_{2}\rightarrow I,~~~as~~\lambda\rightarrow -\infty.
\end{align*}
To solve the RH problem, we consider the following theorem.
\begin{prop}
\begin{align*}
&\det P_{1}=r_{55},~~~~\lambda\in C^{+},\\
&\det P_{2}=s_{55},~~~~\lambda\in C^{-}.
\end{align*}
\end{prop}
\begin{proof}
According to Eq.\eqref{lax-26} and Eq.\eqref{lax-27}, we write $P_{1}$ and $P_{2}$ in the form of
\begin{equation}\label{lax-31}
\begin{split}
&P_{1}=\mu_{-}H_{1}+\mu_{-}H_{2}+\mu_{-}H_{3}+\mu_{-}H_{4}+\mu_{+}H_{5},\\
&P_{2}=H_{1}\mu_{-}^{-1}+H_{2}\mu_{-}^{-1}+H_{3}\mu_{-}^{-1}+H_{4}\mu_{-}^{-1}+H_{5}\mu_{+}^{-1},
\end{split}
\end{equation}
where $H_{1}=\mbox{dig}(1,0,0,0,0)$, $H_{2}=\mbox{dig}(0,1,0,0,0)$, $H_{3}=\mbox{dig}(0,0,1,0,0)$, $H_{4}=\mbox{dig}(0,0,0,1,0)$ and $H_{5}=\mbox{dig}(0,0,0,0,1)$. Hence
\begin{equation}\label{lax-32}
\begin{split}
\det P_{1}(\lambda)&=\det(\mu_{-}H_{1}+\mu_{-}H_{2}+\mu_{-}H_{3}+\mu_{-}H_{4}+\mu_{+}H_{5}) \\
&=\det\mu_{-}\cdot \det(H_{1}+H_{2}+H_{3}+H_{4}+ER(\lambda)E^{-1}H_{5})\\
&=1\cdot \det\left(\begin{array}{ccccc}
   1 & 0 & 0 & 0 & r_{15}e^{-2i\lambda x} \\
  0 &1 & 0 & 0 & r_{25}e^{-2i\lambda x} \\
   0 & 0 & 1 & 0 & r_{35}e^{-2i\lambda x} \\
   0 & 0 & 0 & 1 & r_{45}e^{-2i\lambda x} \\
    0 & 0& 0 & 0 &r_{55}\\
    \end{array}\right)\\
    &=r_{55},
\end{split}
\end{equation}
apply the same method to $P_{2}(\lambda)$,
\begin{equation}\label{lax-34}
\det P_{2}(\lambda)=s_{55}.
\end{equation}
From Eq.\eqref{lax-19}, we know that $\det\mu_{\pm}=1$, according to above analysis, we can get $\det P_{1}=r_{55}, \lambda\in C^{+}$, and $\det P_{2}=s_{55}, \lambda\in C^{-}$.
\end{proof}
As we can see matrix $P$ has the symmetry relation
\begin{equation}\label{lax-35}
P^{\dag}=-BPB^{-1},
\end{equation}
symbol $``\dag"$ represents the Hermitian of a matrix, and
\begin{equation}\label{lax-36}
B=\left(\begin{array}{ccccc}
   a_{11} & a_{21}^{\ast} & a_{31}^{\ast} & a_{41}^{\ast} & 0 \\
  a_{21} &a_{22} & a_{32}^{\ast} & a_{42}^{\ast} & 0 \\
   a_{31} & a_{32} & a_{33} & a_{43}^{\ast} & 0 \\
    a_{41}& a_{42} & a_{43} & a_{44} & 0 \\
    0 & 0 & 0 & 0 & -1\\
    \end{array}\right).
\end{equation}

According to Eq.\eqref{lax-9.1} and Eq.\eqref{lax-35}, $\mu_{\pm}$  meet the following relation
\begin{equation}\label{lax-37}
B^{-1}\mu_{\pm}^{\dagger}(\lambda^{\ast})B=\mu_{\pm}^{-1}(\lambda),
\end{equation}
the scattering matrix $S(\lambda)$ satisfies the equation
\begin{equation}\label{lax-38}
B^{-1}S^{\dagger}(\lambda^{\ast})B=S^{-1}(\lambda)=R(\lambda).
\end{equation}
Eq.\eqref{lax-38} evidently shows
\begin{align}\label{lax-39}
\begin{split}
r_{55}(\lambda)&=s_{55}^{\ast}(\lambda^{\ast}),~~~\lambda\in C^{+},\\
-s_{51}^{*}(\lambda)&=a_{11}r_{15}+a_{21}^{*}r_{25}+a_{31}^{*}r_{35}+a_{41}^{*}r_{45},~~~\lambda\in R,\\
-s_{52}^{*}(\lambda)&=a_{21}r_{15}+a_{22}r_{25}+a_{32}^{*}r_{35}+a_{42}^{*}r_{45},~~~\lambda\in R,\\
-s_{53}^{*}(\lambda)&=a_{31}r_{15}+a_{32}r_{25}+a_{33}r_{35}+a_{43}^{*}r_{45},~~~\lambda\in R,\\
-s_{54}^{*}(\lambda)&=a_{41}r_{15}+a_{42}r_{25}+a_{43}r_{35}+a_{44}r_{45},~~~\lambda\in R.
\end{split}
\end{align}

\begin{thm}
\begin{equation}\label{lax-44}
P_{1}^{\dagger}(\lambda^{\ast})=BP_{2}(\lambda)B^{-1},~~~\lambda\in C^{-}.
\end{equation}
\end{thm}

\begin{proof}
According to Eq.\eqref{lax-31}, we have $P_{1}=\mu_{-}H_{1}+\mu_{-}H_{2}+\mu_{-}H_{3}+\mu_{-}H_{4}+\mu_{+}H_{5}$,
\begin{align}\label{lax-43}
\begin{split}
P_{1}^{\dagger}(\lambda^{\ast})=&(\mu_{-}(\lambda^{\ast})H_{1}+\mu_{-}(\lambda^{\ast})H_{2}+\mu_{-}(\lambda^{\ast})H_{3}+\mu_{-}(\lambda^{\ast})H_{4}+\mu_{+}(\lambda^{\ast})H_{5})^{\dagger}\\
=&H_{1}\mu_{-}^{\dagger}(\lambda^{\ast})+H_{2}\mu_{-}^{\dagger}(\lambda^{\ast})+H_{3}\mu_{-}^{\dagger}(\lambda^{\ast})+H_{4}\mu_{-}^{\dagger}(\lambda^{\ast})+H_{5}\mu_{+}^{\dagger}(\lambda^{\ast})\\
=&BP_{2}(\lambda)B^{-1},~~\lambda\in C^{-}.
\end{split}
\end{align}
So that
\begin{equation*}
P_{1}^{\dagger}(\lambda^{\ast})=BP_{2}(\lambda)B^{-1},~~\lambda\in C^{-}.
\end{equation*}
\end{proof}

From Eq.\eqref{lax-32}, Eq.\eqref{lax-34} and Eq.\eqref{lax-39}, we see $\det P_{1}(\lambda)=(\det P_{2}(\lambda^{*}))^{*}$, if $\det P_{1}$ have a zero $\lambda$, $\det P_{2}$ have a zero $\lambda^{*}$. So we suppose that $\det P_{1}$ has N simple zeros $\{\lambda_{j}\}_{1}^{N}$ in $C^{+}$, and $\det P_{2}$ has N simple zeros $\{\lambda^{*}_{j}\}_{1}^{N}$ in $C^{-}$. These zeros with the nonzero vectors $v_{j}$ and $\hat{v}_{j}$, set up of the full generic discrete data, which satisfy the equations
\begin{align}\label{lax-45}
\begin{split}
&P_{1}(\lambda_{j})v_{j}=0,\\
&\hat{v}_{j}P_{2}(\lambda^{*})=0,
\end{split}
\end{align}
where $v_{j}$ is column vector, and $\hat{v}_{j}$ is the row vector. From Eq.\eqref{lax-44}, Eqs.\eqref{lax-45} one obtains that the eigenvectors admit the following relation.
\begin{equation}\label{lax-47}
\hat{v}_{j}=v_{j}^{\dagger}B,~~~1\leq j \leq N.
\end{equation}
 Then we analyze the time-spatial revolution with $v_{j}$. We take the derivative of the first equation of Eqs.\eqref{lax-45} with respect to x, apply the same method to t.
\begin{align}\label{lax-48}
\begin{split}
&P_{1,x}v_{j}+P_{1}v_{j,x}=0,\\
&P_{1,t}v_{j}+P_{1}v_{j,t}=0.
\end{split}
\end{align}
On the basis of
\begin{align*}\label{lax-50}
P_{1,x}&=(\mu_{-}H_{1}+\mu_{-}H_{2}+\mu_{-}H_{3}+\mu_{+}H_{4})_{x}\\
&=\mu_{-,x}H_{1}+\mu_{-,x}H_{2}+\mu_{-,x}H_{3}+\mu_{+,x}H_{4},\\
\end{align*}
and the Lax pair Eqs.\eqref{lax-9.1}, we have
\begin{equation}\label{lax-51}
\begin{split}
P_{1,x}=&[-i\lambda(\Lambda\mu_{-}-\mu_{-}\Lambda)+P\mu_{-}]H_{1}
+[-i\lambda(\Lambda\mu_{-}-\mu_{-}\Lambda)+P\mu_{-}]H_{2}\\
&+[-i\lambda(\Lambda\mu_{-}-\mu_{-}\Lambda)+P\mu_{-}]H_{3}
+[-i\lambda(\Lambda\mu_{-}-\mu_{-}\Lambda)+P\mu_{-}]H_{4}\\
&+[-i\lambda(\Lambda\mu_{+}-\mu_{+}\Lambda)+P\mu_{+}]H_{5}\\
=&-i\lambda\Lambda P_{1}+i\lambda P_{1}\Lambda+PP_{1}\\
=&-i\lambda[\Lambda,P_{1}]+PP_{1}.
\end{split}
\end{equation}
Applying the same method to $P_{1,t}$, we get
\begin{equation}\label{lax-52}
P_{1,t}=-2i\lambda^{2}[\Lambda,P_{1}]+QP_{1}.
\end{equation}
Inserting Eq.\eqref{lax-51} and Eq.\eqref{lax-52} into the first equation of Eqs.\eqref{lax-48} and the second equation of Eqs.\eqref{lax-48}, respectively. Noticing that $P_{1}v_{j}=0$, we have
\begin{align}\label{lax-53}
\begin{split}
&i\lambda\Lambda v_{j}+v_{j,x}=0,\\
&2i\lambda^{2}\Lambda v_{j}+v_{j,t}=0.
\end{split}
\end{align}
According to Eqs.\eqref{lax-53}, we have
\begin{equation*}\label{lax-55}
v_{j}=e^{-i(\lambda_{j}x+2\lambda_{j}^{2}t)\Lambda}v_{j,0},
\end{equation*}
where $v_{j,0}$ are complex constant vectors. From Eq.\eqref{lax-47}, we have
\begin{equation*}\label{lax-56}
\hat{v}_{j}=v_{j}^{\dagger}(\lambda_{j})B=v^{\dagger}_{j,0}e^{i(\lambda^{\ast}_{j}x+2\lambda^{\ast2}_{j}t)\Lambda}B.
\end{equation*}

\section{ Multi-soliton solutions}
Now, we are going to expand $P_{1}(\lambda)$ at large-$\lambda$ as
\begin{equation}\label{lax-57}
P_{1}(\lambda)=I+\frac{P_{1}^{(1)}}{\lambda}+\frac{P_{1}^{(2)}}{\lambda^{2}}+o(\frac{1}{\lambda^{3}}),~~~\lambda\rightarrow\infty.
\end{equation}
Inserting Eq.\eqref{lax-57} into Eq.\eqref{lax-9.1}
\begin{equation}\label{lax-58}
o(1)~:~i[\Lambda,P_{1}^{(1)}]=P.
\end{equation}
From Eq.\eqref{lax-58}, we can generate
\begin{align}\label{lax-59}
\begin{split}
&q_{1}(x,t)=2i(P_{1}^{(1)})_{15},\\
&q_{2}(x,t)=2i(P_{1}^{(1)})_{25},\\
&q_{3}(x,t)=2i(P_{1}^{(1)})_{35},\\
&q_{4}(x,t)=2i(P_{1}^{(1)})_{45},
\end{split}
\end{align}
where $(P_{1}^{(1)})_{ij}$ is the $(i,j)$-entry of matrix $P_{1}^{(1)}$.

To obtain soliton solutions, we set $G=I$ in \eqref{lax-29}. The solutions for this special RH problem \eqref{lax-29} can be given as
\begin{align}\label{lax-62}
\begin{split}
&P_{1}(\lambda)=I-\sum_{k=1}^{N}\sum_{j=1}^{N}\frac{v_{k}\hat{v}_{j}(M^{-1})_{kj}}{\lambda-\hat{\lambda_{j}}},\\
&P_{2}(\lambda)=I+\sum_{k=1}^{N}\sum_{j=1}^{N}\frac{v_{k}\hat{v}_{j}(M^{-1})_{kj}}{\lambda-\lambda_{j}},
\end{split}
\end{align}
where $M$ is a $N \times N$ matrix with entries
\begin{equation}\label{lax-64}
M_{kj}=\frac{\hat{v}_{k}v_{j}}{\lambda_{j}-\hat{\lambda}_{k}},
\end{equation}
and $(M^{-1})_{kj}$ means the $(k,j)$-entry of the inverse matrix of $M$. From expression Eqs.\eqref{lax-62}, one has
\begin{equation}\label{lax-65}
P_{1}^{(1)}=-\sum_{k=1}^{N}\sum_{j=1}^{N}v_{k}\hat{v}_{j}(M^{-1})_{kj}.
\end{equation}
Then setting nonzero vectors $v_{k,0}=(\alpha_{k},\beta_{k},\tau_{k},\zeta_{k},\gamma_{k})^{T}$ and $\theta_{k}=-i(\lambda_{k}x+2\lambda_{k}^{2}t)$, we generate
\begin{equation}\label{lax-66}
v_{k}=e^{\theta_{k}\Lambda}v_{k,0}=\left(\begin{array}{ccccc}
   e^{\theta_{k}} & 0 & 0 & 0 & 0 \\
  0 &e^{\theta_{k}} & 0 & 0 & 0 \\
   0 & 0 & e^{\theta_{k}} & 0 & 0 \\
   0 & 0 & 0 & e^{\theta_{k}} & 0 \\
    0 & 0 & 0 & 0 &e^{-\theta_{k}}\\
    \end{array}\right)\left(\begin{array}{c}
   \alpha_{k}  \\
  \beta_{k}  \\
   \tau_{k}  \\
   \zeta_{k}  \\
    \gamma_{k} \\
    \end{array}\right)=\left(\begin{array}{c}
   \alpha_{k}e^{\theta_{k}}  \\
  \beta_{k}e^{\theta_{k}}  \\
   \tau_{k}e^{\theta_{k}}  \\
   \zeta_{k}e^{\theta_{k}}  \\
    \gamma_{k}e^{-\theta_{k}} \\
    \end{array}\right),
\end{equation}
\begin{equation}\label{lax-67}
\begin{split}
\hat{v}_{j}=&v_{j}^{\dagger}(\lambda_{j})B\\
=&(\alpha_{k}e^{\theta_{k}},\beta_{k}e^{\theta_{k}},\tau_{k}e^{\theta_{k}},\zeta_{k}e^{\theta_{k}},\gamma_{k}e^{-\theta_{k}})
    \left(\begin{array}{ccccc}
   a_{11} & a_{21}^{\ast} & a_{31}^{\ast} & a_{41}^{\ast} & 0 \\
  a_{21} &a_{22} & a_{32}^{\ast} & a_{42}^{\ast} & 0 \\
   a_{31} & a_{32} & a_{33} & a_{43}^{\ast} & 0 \\
    a_{41}& a_{42} & a_{43} & a_{44} & 0 \\
    0 & 0 & 0 & 0 & -1\\
    \end{array}\right)\\
=&
   (a_{11}\alpha_{j}^{\ast}e^{\theta_{j}^{\ast}}+a_{21}\beta_{j}^{\ast}e^{\theta_{j}^{\ast}}+a_{31}\tau_{j}^{\ast}e^{\theta_{j}^{\ast}}+a_{41}\zeta_{j}^{\ast}e^{\theta_{j}^{\ast}},a_{21}^{\ast}\alpha_{j}^{\ast}e^{\theta_{j}^{\ast}}+a_{22}\beta_{j}^{\ast}e^{\theta_{j}^{\ast}}+a_{32}\tau_{j}^{\ast}e^{\theta_{j}^{\ast}}\\&+a_{42}\zeta_{j}^{\ast}e^{\theta_{j}^{\ast}}, a_{31}^{\ast}\alpha_{j}^{\ast}e^{\theta_{j}^{\ast}}+a_{32}^{\ast}\beta_{j}^{\ast}e^{\theta_{j}^{\ast}}+a_{33}\tau_{j}^{\ast}e^{\theta_{j}^{\ast}}+a_{43}\zeta_{j}^{\ast}e^{\theta_{j}^{\ast}},a_{41}^{\ast}\alpha_{j}^{\ast}e^{\theta_{j}^{\ast}}+a_{42}^{\ast}\beta_{j}^{\ast}e^{\theta_{j}^{\ast}}\\&+a_{43}^{\ast}\tau_{j}^{\ast}e^{\theta_{j}^{\ast}}+a_{44}\zeta_{j}^{\ast}e^{\theta_{j}^{\ast}}, -\gamma_{j}^{\ast}e^{-\theta_{j}^{\ast}}).
\end{split}
\end{equation}
Obviously
\begin{equation}\label{lax-68}
v_{k}\hat{v}_{j}=\left(\begin{array}{ccccc}
   b_{11}&b_{12} &b_{13} &b_{14} &b_{15}\\
   b_{21}&b_{22} &b_{23} &b_{24} &b_{25}\\
   b_{31}&b_{32} &b_{33} &b_{34} &b_{35}\\
   b_{41}&b_{42} &b_{43} &b_{44} &b_{45}\\
   b_{51}&b_{52} &b_{53} &b_{54} &b_{55}\\
    \end{array}\right),
\end{equation}
\begin{equation}\label{lax-69}
\begin{split}
\hat{v}_{k}v_{j}=&(a_{11}\alpha_{k}^{\ast}\alpha_{j}+a_{21}\beta_{k}^{\ast}\alpha_{j}+a_{31}\tau_{k}^{\ast}\alpha_{j}+a_{41}\zeta_{k}^{\ast}\alpha_{j}+a_{21}^{\ast}\alpha_{k}^{\ast}\beta_{j}+a_{22}\beta_{k}^{\ast}\beta_{j}+a_{32}\tau_{k}^{\ast}\beta_{j}\\
&+a_{42}\zeta_{k}^{\ast}\beta_{j}+a_{31}^{\ast}\alpha_{k}^{\ast}\tau_{j}+a_{32}^{\ast}\beta_{k}^{\ast}\tau_{j}+a_{33}\tau_{k}^{\ast}\tau_{j}+a_{43}\zeta_{k}^{\ast}\tau_{j}+a_{41}^{\ast}\alpha_{k}^{\ast}\zeta_{j}+a_{42}^{\ast}\beta_{k}^{\ast}\zeta_{j}\\
&+a_{43}^{\ast}\tau_{k}^{\ast}\zeta_{j}+a_{44}\zeta_{k}^{\ast}\zeta_{j})e^{\theta_{k}^{\ast}+\theta_{j}}-\gamma_{k}^{\ast}\gamma_{j}e^{-\theta_{k}^{\ast}-\theta_{j}}.
\end{split}
\end{equation}
\begin{equation*}
\begin{split}
&b_{15}=-\alpha_{k}\gamma_{j}^{\ast}e^{\theta_{k}-\theta_{j}^{\ast}},~~~~b_{25}=-\beta_{k}\gamma_{j}^{\ast}e^{\theta_{k}-\theta_{j}^{\ast}}~~~~b_{35}=-\tau_{k}\gamma_{j}^{\ast}e^{\theta_{k}-\theta_{j}^{\ast}},~~~~b_{45}=-\zeta_{k}\gamma_{j}^{\ast}e^{\theta_{k}-\theta_{j}^{\ast}},\\
&b_{55}=-\gamma_{k}\gamma_{j}^{\ast}e^{\theta_{k}-\theta_{j}^{\ast}}.
\end{split}
\end{equation*}
It should be noted that the parameter $b_{ij} ~(i\leq5,~j\leq4)$ do not work on the construction solutions, so the specific expression is not given for convenience.

As a consequence, general $N$-soliton solution for the FCNLS equation \eqref{NLS-1} can be derived as follows
\begin{align}\label{lax-70}
\begin{split}
&q_{1}=2i\sum_{k=1}^{N}\sum_{j=1}^{N}\alpha_{k}\gamma_{j}^{\ast}e^{\theta_{k}-\theta_{j}^{\ast}}(M^{-1})_{kj},\\
&q_{2}=2i\sum_{k=1}^{N}\sum_{j=1}^{N}\beta_{k}\gamma_{j}^{\ast}e^{\theta_{k}-\theta_{j}^{\ast}}(M^{-1})_{kj},\\
&q_{3}=2i\sum_{k=1}^{N}\sum_{j=1}^{N}\tau_{k}\gamma_{j}^{\ast}e^{\theta_{k}-\theta_{j}^{\ast}}(M^{-1})_{kj},\\
&q_{4}=2i\sum_{k=1}^{N}\sum_{j=1}^{N}\zeta_{k}\gamma_{j}^{\ast}e^{\theta_{k}-\theta_{j}^{\ast}}(M^{-1})_{kj},
\end{split}
\end{align}
where
\begin{equation}\label{lax-73}
\begin{split}
M_{kj}=&\frac{1}{\lambda_{j}-\hat{\lambda}_{k}}[(a_{11}\alpha_{k}^{\ast}\alpha_{j}+a_{21}\beta_{k}^{\ast}\alpha_{j}+a_{31}\tau_{k}^{\ast}\alpha_{j}+a_{41}\zeta_{k}^{\ast}\alpha_{j}+a_{21}^{\ast}\alpha_{k}^{\ast}\beta_{j}+a_{22}\beta_{k}^{\ast}\beta_{j}\\
&+a_{32}\tau_{k}^{\ast}\beta_{j}+a_{42}\zeta_{k}^{\ast}\beta_{j}+a_{31}^{\ast}\alpha_{k}^{\ast}\tau_{j}+a_{32}^{\ast}\beta_{k}^{\ast}\tau_{j}+a_{33}\tau_{k}^{\ast}\tau_{j}+a_{43}\zeta_{k}^{\ast}\tau_{j}+a_{41}^{\ast}\alpha_{k}^{\ast}\zeta_{j}\\
&+a_{42}^{\ast}\beta_{k}^{\ast}\zeta_{j}+a_{43}^{\ast}\tau_{k}^{\ast}\zeta_{j}+a_{44}\zeta_{k}^{\ast}\zeta_{j})e^{\theta_{k}^{\ast}+\theta_{j}}-\gamma_{k}^{\ast}\gamma_{j}e^{-\theta_{k}^{\ast}-\theta_{j}}],
1\leq k,j \leq N.
\end{split}
\end{equation}
To make the expression \eqref{lax-70} simpler,  we define the following matrix $F$, $G$, $H$ and $K$.
\begin{equation}\label{lax-74}
F=\left(\begin{array}{ccccc}
   0 & \alpha_{1}e^{\theta_{1}} & \alpha_{2}e^{\theta_{2}} & \ldots &\alpha_{N}e^{\theta_{N}}\\
  \gamma_{1}^{\ast}e^{-\theta_{1}^{\ast}} &M_{11} & M_{12} & \ldots &M_{1N}\\
   \gamma_{2}^{\ast}e^{-\theta_{2}^{\ast}} & M_{21} & M_{22} & \ldots &M_{2N}\\
    \vdots & \vdots & \vdots &\vdots&\vdots\\
   \gamma_{N}^{\ast}e^{-\theta_{N}^{\ast}} & M_{N1} & M_{N2} &\ldots&M_{NN}\\
    \end{array}\right),
\end{equation}
\begin{equation}\label{lax-75}
G=\left(\begin{array}{ccccc}
   0 & \beta_{1}e^{\theta_{1}} & \beta_{2}e^{\theta_{2}} & \ldots &\beta_{N}e^{\theta_{N}}\\
  \gamma_{1}^{\ast}e^{-\theta_{1}^{\ast}} &M_{11} & M_{12} & \ldots &M_{1N}\\
   \gamma_{2}^{\ast}e^{-\theta_{2}^{\ast}} & M_{21} & M_{22} & \ldots &M_{2N}\\
    \vdots & \vdots & \vdots &\vdots&\vdots\\
   \gamma_{N}^{\ast}e^{-\theta_{N}^{\ast}} & M_{N1} & M_{N2} &\ldots&M_{NN}\\
    \end{array}\right),
\end{equation}
\begin{equation}\label{lax-76}
H=\left(\begin{array}{ccccc}
   0 & \tau_{1}e^{\theta_{1}} & \tau_{2}e^{\theta_{2}} & \ldots &\tau_{N}e^{\theta_{N}}\\
  \gamma_{1}^{\ast}e^{-\theta_{1}^{\ast}} &M_{11} & M_{12} & \ldots &M_{1N}\\
   \gamma_{2}^{\ast}e^{-\theta_{2}^{\ast}} & M_{21} & M_{22} & \ldots &M_{2N}\\
    \vdots & \vdots & \vdots &\vdots&\vdots\\
   \gamma_{N}^{\ast}e^{-\theta_{N}^{\ast}} & M_{N1} & M_{N2} &\ldots&M_{NN}\\
    \end{array}\right).
\end{equation}
\begin{equation}\label{lax-761}
K=\left(\begin{array}{ccccc}
   0 & \zeta_{1}e^{\theta_{1}} & \zeta_{2}e^{\theta_{2}} & \ldots &\zeta_{N}e^{\theta_{N}}\\
  \gamma_{1}^{\ast}e^{-\theta_{1}^{\ast}} &M_{11} & M_{12} & \ldots &M_{1N}\\
   \gamma_{2}^{\ast}e^{-\theta_{2}^{\ast}} & M_{21} & M_{22} & \ldots &M_{2N}\\
    \vdots & \vdots & \vdots &\vdots&\vdots\\
   \gamma_{N}^{\ast}e^{-\theta_{N}^{\ast}} & M_{N1} & M_{N2} &\ldots&M_{NN}\\
    \end{array}\right).
\end{equation}
 On the basis of matrix \eqref{lax-74}, \eqref{lax-75}, \eqref{lax-76} and \eqref{lax-761}, we can get
\begin{equation}\label{lax-77}
q_{1}=-2i\frac{\det F}{\det M},~~q_{2}=-2i\frac{\det G}{\det M},~~q_{3}=-2i\frac{\det H}{\det M},~~q_{4}=-2i\frac{\det K}{\det M}.
\end{equation}

In the remainder of this section, we figure out the soliton solutions in the case of $N=1$ and $N=2$. In the case of $N=1$, we obtain the one-soliton solution
\begin{align}\label{lax-78}
\begin{split}
&q_{1}=2i\frac{\alpha_{1}\gamma_{1}^{\ast}e^{\theta_{1}-\theta_{1}^{\ast}}}{M_{11}},\\
&q_{2}=2i\frac{\beta_{1}\gamma_{1}^{\ast}e^{\theta_{1}-\theta_{1}^{\ast}}}{M_{11}},\\
&q_{3}=2i\frac{\tau_{1}\gamma_{1}^{\ast}e^{\theta_{1}-\theta_{1}^{\ast}}}{M_{11}},\\
&q_{4}=2i\frac{\zeta_{1}\gamma_{1}^{\ast}e^{\theta_{1}-\theta_{1}^{\ast}}}{M_{11}},
\end{split}
\end{align}
\begin{equation*}
\begin{split}
M_{11}=&\frac{1}{\lambda_{1}-\hat{\lambda}_{1}}[(a_{11}|\alpha_{1}|^{2}+a_{22}|\beta_{1}|^{2}+a_{33}|\tau_{1}|^{2}+a_{44}|\zeta_{1}|^{2}+a_{21}\beta_{1}^{\ast}\alpha_{1}+a_{31}\tau_{1}^{\ast}\alpha_{1}\\
&+a_{41}\zeta_{1}^{\ast}\alpha_{1}+a_{21}^{\ast}\alpha_{1}^{\ast}\beta_{1}+a_{32}\tau_{1}^{\ast}\beta_{1}+a_{42}\zeta_{1}^{\ast}\beta_{1}+a_{31}^{\ast}\alpha_{1}^{\ast}\tau_{1}+a_{32}^{\ast}\beta_{1}^{\ast}\tau_{1}+a_{43}\zeta_{1}^{\ast}\tau_{1}\\
&+a_{41}^{\ast}\alpha_{1}^{\ast}\zeta_{1}+a_{42}^{\ast}\beta_{1}^{\ast}\zeta_{1}+a_{43}^{\ast}\tau_{1}^{\ast}\zeta_{1})e^{\theta_{1}+\theta_{1}^{\ast}}-|\gamma_{1}|^{2}e^{-(\theta_{1}+\theta_{1}^{\ast})}],
\end{split}
\end{equation*}
where $\theta_{1}=-i(\lambda_{1}x+2\lambda_{1}^{2}t)$.  Furthermore, fixing $\gamma_{1}=1$, $\lambda_{1}=n_{1}+im_{1}$ and setting
\begin{equation}\label{111}
\begin{split}
-&(a_{11}|\alpha_{1}|^{2}+a_{22}|\beta_{1}|^{2}+a_{33}|\tau_{1}|^{2}+a_{44}|\zeta_{1}|^{2}+a_{21}\beta_{1}^{\ast}\alpha_{1}+a_{31}\tau_{1}^{\ast}\alpha_{1}
+a_{41}\zeta_{1}^{\ast}\alpha_{1}\\
&+a_{21}^{\ast}\alpha_{1}^{\ast}\beta_{1}+a_{32}\tau_{1}^{\ast}\beta_{1}+a_{42}\zeta_{1}^{\ast}\beta_{1}+a_{31}^{\ast}\alpha_{1}^{\ast}\tau_{1}+a_{32}^{\ast}\beta_{1}^{\ast}\tau_{1}+a_{43}\zeta_{1}^{\ast}\tau_{1}
+a_{41}^{\ast}\alpha_{1}^{\ast}\zeta_{1}\\
&+a_{42}^{\ast}\beta_{1}^{\ast}\zeta_{1}+a_{43}^{\ast}\tau_{1}^{\ast}\zeta_{1})=e^{2\xi_{1}}, \end{split}
\end{equation}
the expression \eqref{lax-78} can be written as
\begin{equation}\label{lax-81}
\begin{split}
&q_{1}=2m_{1}\alpha_{1}\gamma_{1}^{\ast}e^{-\xi_{1}}e^{\theta_{1}-\theta_{1}^{\ast}}\mbox{sech}(\theta_{1}^{\ast}+\theta_{1}+\xi_{1}),\\
&q_{2}=2m_{1}\beta_{1}\gamma_{1}^{\ast}e^{-\xi_{1}}e^{\theta_{1}-\theta_{1}^{\ast}}\mbox{sech}(\theta_{1}^{\ast}+\theta_{1}+\xi_{1}),\\
&q_{3}=2m_{1}\tau_{1}\gamma_{1}^{\ast}e^{-\xi_{1}}e^{\theta_{1}-\theta_{1}^{\ast}}\mbox{sech}(\theta_{1}^{\ast}+\theta_{1}+\xi_{1}),\\
&q_{4}=2m_{1}\zeta_{1}\gamma_{1}^{\ast}e^{-\xi_{1}}e^{\theta_{1}-\theta_{1}^{\ast}}\mbox{sech}(\theta_{1}^{\ast}+\theta_{1}+\xi_{1}).
\end{split}
\end{equation}
According to the notation above, we have
\begin{equation}\label{lax-84}
\begin{split}
&\theta_{1}-\theta_{1}^{\ast}=-2in_{1}x-4in_{1}^{2}t,\\
&\theta_{1}+\theta_{1}^{\ast}=2m_{1}x+4im_{1}^{2}t+8n_{1}m_{1}t.
\end{split}
\end{equation}
Thus the one-soliton solutions in \eqref{lax-78} can be further written as
\begin{align}\label{lax-86}
\begin{split}
&q_{1}=2m_{1}\alpha_{1}\gamma_{1}^{\ast}e^{-\xi_{1}}e^{-2in_{1}x-4in_{1}^{2}t}
\mbox{sech}(2m_{1}x+4im_{1}^{2}t+8n_{1}m_{1}t+\xi_{1}),\\
&q_{2}=2m_{1}\beta_{1}\gamma_{1}^{\ast}e^{-\xi_{1}}e^{-2in_{1}x-4in_{1}^{2}t}
\mbox{sech}(2m_{1}x+4im_{1}^{2}t+8n_{1}m_{1}t+\xi_{1}),\\
&q_{3}=2m_{1}\tau_{1}\gamma_{1}^{\ast}e^{-\xi_{1}}e^{-2in_{1}x-4in_{1}^{2}t}
\mbox{sech}(2m_{1}x+4im_{1}^{2}t+8n_{1}m_{1}t+\xi_{1}),\\
&q_{4}=2m_{1}\zeta_{1}\gamma_{1}^{\ast}e^{-\xi_{1}}e^{-2in_{1}x-4in_{1}^{2}t}
\mbox{sech}(2m_{1}x+4im_{1}^{2}t+8n_{1}m_{1}t+\xi_{1}).
\end{split}
\end{align}
From the Eq.\eqref{lax-86}, we can know that the one-soliton solutions of $q_{1}$, $q_{2}$, $q_{3}$ and $q_{4}$ can be described by hyperbolic cosecant function. Taking an example, $q_{1}$ has the peak amplitude
\begin{equation*}\label{lax-89}
\Upsilon_{1}=2m_{1}\alpha_{1}e^{-\xi_{1}},
\end{equation*}
and the velocity
\begin{equation*}\label{lax-90}
\varpi_{1}=2im_{1}t+4n_{1}t,
\end{equation*}
Similarly, we can know the peak amplitude and velocity of $q_{2}$, $q_{3}$ and $q_{4}$, respectively.
\begin{equation*}\label{lax-91}
\Upsilon_{2}=2m_{1}\beta_{1}e^{-\xi_{1}},~~~\Upsilon_{3}=2m_{1}\tau_{1}e^{-\xi_{1}},~~~\Upsilon_{4}=2m_{1}\zeta_{1}e^{-\xi_{1}},
\end{equation*}
\begin{equation*}\label{lax-92}
\varpi_{2}=2im_{1}t+4n_{1}t,~~~\varpi_{3}=2im_{1}t+4n_{1}t,~~~\varpi_{4}=2im_{1}t+4n_{1}t.
\end{equation*}
From the expressions of $\Upsilon_{1}$, $\Upsilon_{2}$, $\Upsilon_{3}$ and $\Upsilon_{4}$, $\varpi_{1}$, $\varpi_{2}$, $\varpi_{3}$, $\varpi_{4}$ we can know that all of them rely on both the real part $n_{1}$ and the imaginary part
$m_{1}$ of the eigenvalue $\lambda_{1}$. Figure 1, Figure 2, Figure 3 and Figure 4 represent the localized structures and
dynamic behaviors of the single-soliton solution. All the analysis of $q_{1}$ be the same with $q_{2}$, $q_{3}$ and $q_{4}$.

{\rotatebox{0}{\includegraphics[width=3.6cm,height=3.0cm,angle=0]{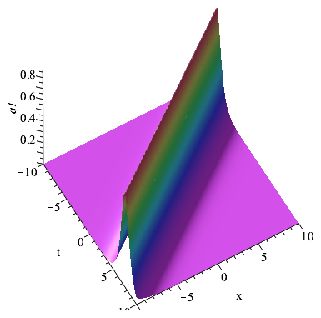}}}
{\rotatebox{0}{\includegraphics[width=3.3cm,height=2.55cm,angle=0]{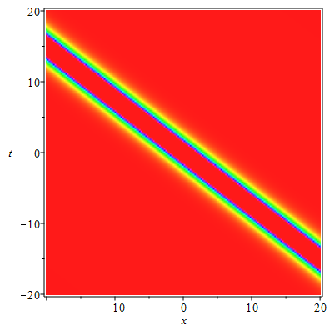}}}
\quad
{\rotatebox{0}{\includegraphics[width=3.6cm,height=2.75cm,angle=0]{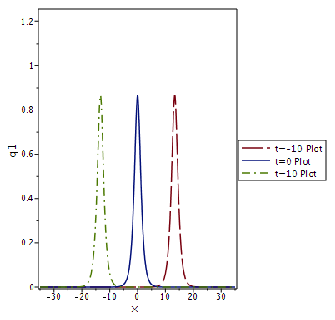}}}

$\quad\qquad\quad(\textbf{a})\quad \ \quad\qquad\qquad\qquad\qquad(\textbf{b})
\quad\qquad\qquad\quad\qquad(\textbf{c})$\\

\noindent { \small \textbf{Figure 1.} Plots of the single-soliton solution $q_{1}$, with the parameters chosen as $a_{11}=a_{22}=a_{33}=a_{44}=0$, $a_{21}=a_{31}=a_{41}=a_{32}=a_{42}=a_{43}=-\frac{1}{9}$, $\alpha_{1}=\tau_{1}=\beta_{1}=\zeta_{1}=\frac{1}{2}-\frac{\sqrt{2}}{2}i$, $\gamma_{1}=1$, $n_{1}=\frac{1}{3}$, $m_{1}=\frac{1}{2}$. $\textbf{(a)}$ three dimensional plot at time $t=0$, $\textbf{(b)}$ density plot, $\textbf{(c)}$  the wave propagation along the $x$-axis with different time.}

{\rotatebox{0}{\includegraphics[width=3.6cm,height=3.0cm,angle=0]{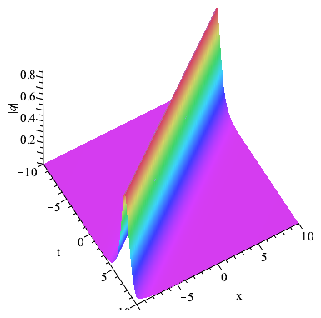}}}
{\rotatebox{0}{\includegraphics[width=3.3cm,height=2.55cm,angle=0]{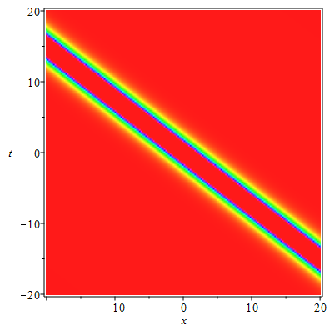}}}
\quad
{\rotatebox{0}{\includegraphics[width=3.6cm,height=2.75cm,angle=0]{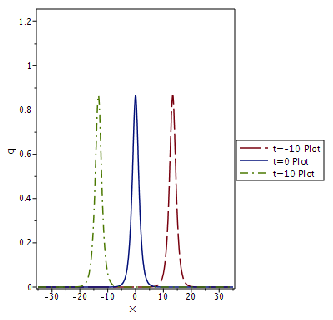}}}

$\quad\qquad\quad(\textbf{a})\quad \ \quad\qquad\qquad\qquad\qquad(\textbf{b})
\quad\qquad\qquad\quad\qquad(\textbf{c})$\\

\noindent { \small \textbf{Figure 2.} Plots of the single-soliton solution $q_{2}$, with the parameters chosen as $a_{11}=a_{22}=a_{33}=a_{44}=0$, $a_{21}=a_{31}=a_{41}=a_{32}=a_{42}=a_{43}=-\frac{1}{9}$, $\alpha_{1}=\tau_{1}=\beta_{1}=\zeta_{1}=\frac{1}{2}-\frac{\sqrt{2}}{2}i$, $\gamma_{1}=1$, $n_{1}=\frac{1}{3}$, $m_{1}=\frac{1}{2}$. $\textbf{(a)}$ three dimensional plot at time $t=0$, $\textbf{(b)}$ density plot, $\textbf{(c)}$  the wave propagation along the $x$-axis with different time.}

{\rotatebox{0}{\includegraphics[width=3.6cm,height=3.0cm,angle=0]{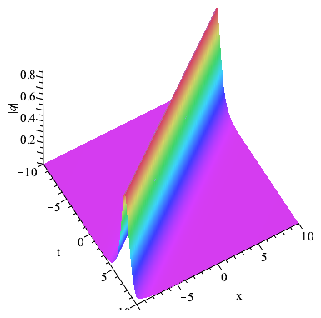}}}
{\rotatebox{0}{\includegraphics[width=3.3cm,height=2.55cm,angle=0]{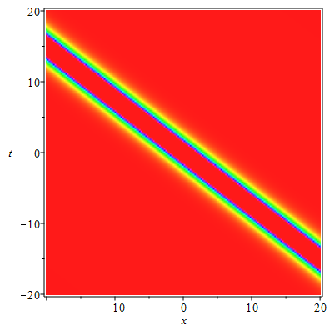}}}
\quad
{\rotatebox{0}{\includegraphics[width=3.6cm,height=2.75cm,angle=0]{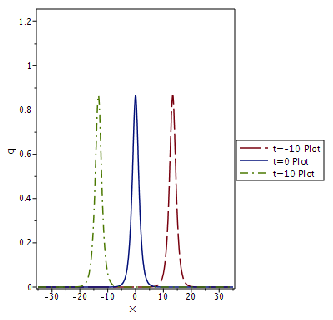}}}

$\quad\qquad\quad(\textbf{a})\quad \ \quad\qquad\qquad\qquad\qquad(\textbf{b})
\quad\qquad\qquad\quad\qquad(\textbf{c})$\\

\noindent { \small \textbf{Figure 3.} Plots of the single-soliton solution $q_{3}$, with the parameters chosen as $a_{11}=a_{22}=a_{33}=a_{44}=0$, $a_{21}=a_{31}=a_{41}=a_{32}=a_{42}=a_{43}=-\frac{1}{9}$, $\alpha_{1}=\tau_{1}=\beta_{1}=\zeta_{1}=\frac{1}{2}-\frac{\sqrt{2}}{2}i$, $\gamma_{1}=1$, $n_{1}=\frac{1}{3}$, $m_{1}=\frac{1}{2}$. $\textbf{(a)}$ three dimensional plot at time $t=0$, $\textbf{(b)}$ density plot, $\textbf{(c)}$  the wave propagation along the $x$-axis with different time.}

{\rotatebox{0}{\includegraphics[width=3.6cm,height=3.0cm,angle=0]{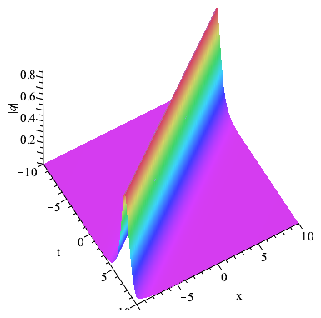}}}
{\rotatebox{0}{\includegraphics[width=3.3cm,height=2.55cm,angle=0]{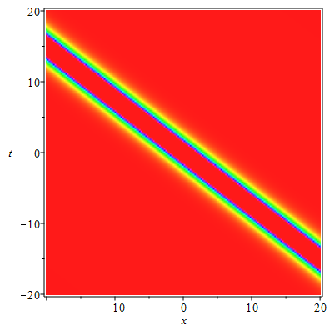}}}
\quad
{\rotatebox{0}{\includegraphics[width=3.6cm,height=2.75cm,angle=0]{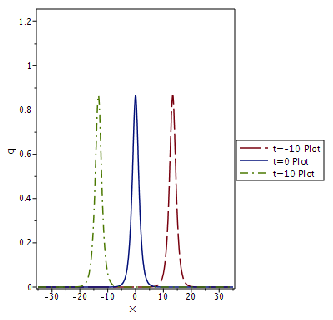}}}

$\quad\qquad\quad(\textbf{a})\quad \ \quad\qquad\qquad\qquad\qquad(\textbf{b})
\quad\qquad\qquad\quad\qquad(\textbf{c})$\\

\noindent { \small \textbf{Figure 4.} Plots of the single-soliton solution $q_{3}$, with the parameters chosen as $a_{11}=a_{22}=a_{33}=a_{44}=0$, $a_{21}=a_{31}=a_{41}=a_{32}=a_{42}=a_{43}=-\frac{1}{9}$, $\alpha_{1}=\tau_{1}=\beta_{1}=\zeta_{1}=\frac{1}{2}-\frac{\sqrt{2}}{2}i$, $\gamma_{1}=1$, $n_{1}=\frac{1}{3}$, $m_{1}=\frac{1}{2}$. $\textbf{(a)}$ three dimensional plot at time $t=0$, $\textbf{(b)}$ density plot, $\textbf{(c)}$  the wave propagation along the $x$-axis with different time.}

In the case of $N=2$,  the two-soliton solutions can be obtained as
\begin{equation}\label{lax-93}
\begin{split}
q_{1}=&\frac{-2i}{M_{11}M_{22}-M_{12}M_{21}}(-\alpha_{1}\gamma_{1}^{\ast}e^{\theta_{1}
-\theta_{1}^{\ast}}M_{22}+\alpha_{1}\gamma_{2}^{\ast}e^{\theta_{1}-\theta_{2}^{\ast}}M_{12}\\
&+\alpha_{2}\gamma_{1}^{\ast}e^{\theta_{2}-\theta_{1}^{\ast}}M_{21}-\alpha_{2}\gamma_{2}^{\ast}
e^{\theta_{2}-\theta_{2}^{\ast}}M_{11}),\\
q_{2}=&\frac{-2i}{M_{11}M_{22}-M_{12}M_{21}}(-\beta_{1}\gamma_{1}^{\ast}e^{\theta_{1}
-\theta_{1}^{\ast}}M_{22}+\beta_{1}\gamma_{2}^{\ast}e^{\theta_{1}-\theta_{2}^{\ast}}M_{12}\\
&+\beta_{2}\gamma_{1}^{\ast}e^{\theta_{2}-\theta_{1}^{\ast}}M_{21}
-\beta_{2}\gamma_{2}^{\ast}e^{\theta_{2}-\theta_{2}^{\ast}}M_{11}),\\
q_{3}=&\frac{-2i}{M_{11}M_{22}-M_{12}M_{21}}(-\tau_{1}\gamma_{1}^{\ast}e^{\theta_{1}-\theta_{1}^{\ast}}
M_{22}+\tau_{1}\gamma_{2}^{\ast}e^{\theta_{1}-\theta_{2}^{\ast}}M_{12}\\
&+\tau_{2}\gamma_{1}^{\ast}e^{\theta_{2}-\theta_{1}^{\ast}}M_{21}
-\tau_{2}\gamma_{2}^{\ast}e^{\theta_{2}-\theta_{2}^{\ast}}M_{11}),\\
q_{4}=&\frac{-2i}{M_{11}M_{22}-M_{12}M_{21}}(-\zeta_{1}\gamma_{1}^{\ast}e^{\theta_{1}
-\theta_{1}^{\ast}}M_{22}+\zeta_{1}\gamma_{2}^{\ast}e^{\theta_{1}-\theta_{2}^{\ast}}M_{12}\\
&+\zeta_{2}\gamma_{1}^{\ast}e^{\theta_{2}-\theta_{1}^{\ast}}M_{21}
-\zeta_{2}\gamma_{2}^{\ast}e^{\theta_{2}-\theta_{2}^{\ast}}M_{11}),
\end{split}
\end{equation}
where
\begin{equation*}\label{lax-96}
\left\{
\begin{aligned}
\begin{split}
M_{11}=&\frac{1}{\lambda_{1}-\lambda_{1}^{\ast}}[(a_{11}|\alpha_{1}|^{2}+a_{22}|\beta_{1}|^{2}+a_{33}|\tau_{1}|^{2}+a_{44}|\zeta_{1}|^{2}+a_{21}\beta_{1}^{\ast}\alpha_{1}+a_{31}\tau_{1}^{\ast}\alpha_{1}\\
&+a_{41}\zeta_{1}^{\ast}\alpha_{1}+a_{21}^{\ast}\alpha_{1}^{\ast}\beta_{1}+a_{32}\tau_{1}^{\ast}\beta_{1}+a_{42}\zeta_{1}^{\ast}\beta_{1}+a_{31}^{\ast}\alpha_{1}^{\ast}\tau_{1}+a_{32}^{\ast}\beta_{1}^{\ast}\tau_{1}+a_{43}\zeta_{1}^{\ast}\tau_{1}\\
&+a_{41}^{\ast}\alpha_{1}^{\ast}\zeta_{1}+a_{42}^{\ast}\beta_{1}^{\ast}\zeta_{1}+a_{43}^{\ast}\tau_{1}^{\ast}\zeta_{1})e^{\theta_{1}^{\ast}+\theta_{1}}-|\gamma_{1}|^{2}e^{-(\theta_{1}^{\ast}+\theta_{1})}],\\
M_{12}=&\frac{1}{\lambda_{2}-\lambda_{1}^{\ast}}[(a_{11}\alpha_{1}^{\ast}\alpha_{2}+a_{21}\beta_{1}^{\ast}\alpha_{2}+a_{31}\tau_{1}^{\ast}\alpha_{2}+a_{41}\zeta_{1}^{\ast}\alpha_{2}+a_{21}^{\ast}\alpha_{1}^{\ast}\beta_{2}+a_{22}\beta_{1}^{\ast}\beta_{2}\\
&+a_{32}\tau_{1}^{\ast}\beta_{2}+a_{42}\zeta_{1}^{\ast}\beta_{2}+a_{31}^{\ast}\alpha_{1}^{\ast}\tau_{2}+a_{32}^{\ast}\beta_{1}^{\ast}\tau_{2}+a_{33}\tau_{1}^{\ast}\tau_{2}+a_{43}\zeta_{1}^{\ast}\tau_{2}+a_{41}^{\ast}\alpha_{1}^{\ast}\zeta_{2}\\
&+a_{42}^{\ast}\beta_{1}^{\ast}\zeta_{2}+a_{43}^{\ast}\tau_{1}^{\ast}\zeta_{2}+a_{44}\zeta_{1}^{\ast}\zeta_{2})e^{\theta_{1}^{\ast}+\theta_{2}}-\gamma_{1}^{\ast}\gamma_{2}e^{-\theta_{1}^{\ast}-\theta_{2}}],\\
M_{21}=&\frac{1}{\lambda_{1}-\lambda_{2}^{\ast}}[(a_{11}\alpha_{2}^{\ast}\alpha_{1}+a_{21}\beta_{2}^{\ast}\alpha_{1}+a_{31}\tau_{2}^{\ast}\alpha_{1}+a_{41}\zeta_{2}^{\ast}\alpha_{1}+a_{21}^{\ast}\alpha_{2}^{\ast}\beta_{1}+a_{22}\beta_{2}^{\ast}\beta_{1}\\
&+a_{32}\tau_{2}^{\ast}\beta_{1}+a_{42}\zeta_{2}^{\ast}\beta_{1}+a_{31}^{\ast}\alpha_{2}^{\ast}\tau_{1}+a_{32}^{\ast}\beta_{2}^{\ast}\tau_{1}+a_{33}\tau_{2}^{\ast}\tau_{1}+a_{43}\zeta_{2}^{\ast}\tau_{1}+a_{41}^{\ast}\alpha_{2}^{\ast}\zeta_{1}\\
&+a_{42}^{\ast}\beta_{2}^{\ast}\zeta_{1}+a_{43}^{\ast}\tau_{2}^{\ast}\zeta_{1}+a_{44}\zeta_{2}^{\ast}\zeta_{1})e^{\theta_{2}^{\ast}+\theta_{1}}-\gamma_{2}^{\ast}\gamma_{1}e^{-\theta_{2}^{\ast}-\theta_{1}}],\\
M_{22}=&\frac{1}{\lambda_{2}-\lambda_{2}^{\ast}}[(a_{11}|\alpha_{2}|^{2}+a_{22}|\beta_{2}|^{2}+a_{33}|\tau_{2}|^{2}+a_{44}|\zeta_{2}|^{2}+a_{21}\beta_{2}^{\ast}\alpha_{2}+a_{31}\tau_{2}^{\ast}\alpha_{2}\\
&+a_{41}\zeta_{2}^{\ast}\alpha_{2}+a_{21}^{\ast}\alpha_{2}^{\ast}\beta_{2}+a_{32}\tau_{2}^{\ast}\beta_{2}+a_{42}\zeta_{2}^{\ast}\beta_{2}+a_{31}^{\ast}\alpha_{2}^{\ast}\tau_{2}+a_{32}^{\ast}\beta_{2}^{\ast}\tau_{2}+a_{43}\zeta_{2}^{\ast}\tau_{2}\\
&+a_{41}^{\ast}\alpha_{2}^{\ast}\zeta_{2}+a_{42}^{\ast}\beta_{2}^{\ast}\zeta_{2}+a_{43}^{\ast}\tau_{2}^{\ast}\zeta_{2})e^{\theta_{2}^{\ast}+\theta_{2}}-|\gamma_{2}|^{2}e^{-(\theta_{2}^{\ast}+\theta_{2})}],\\
\end{split}
\end{aligned}
\right.
\end{equation*}
$\theta_{1}=-i(\lambda_{1}x+2\lambda_{1}^{2}t)$, $\theta_{2}=-i(\lambda_{2}x+2\lambda_{2}^{2}t)$, $\lambda_{1}=n_{1}+im_{1}$ and $\lambda_{2}=n_{2}+im_{2}$. If we let $\gamma_{1}=\gamma_{2}=1$, $\alpha_{1}=\alpha_{2}$, $\beta_{1}=\beta_{2}$, $\tau_{1}=\tau_{2}$, $\zeta_{1}=\zeta_{2}$ and $-(a_{11}|\alpha_{1}|^{2}+a_{22}|\beta_{1}|^{2}+a_{33}|\tau_{1}|^{2}+a_{44}|\zeta_{1}|^{2}+a_{21}\beta_{1}^{\ast}\alpha_{1}+a_{31}\tau_{1}^{\ast}\alpha_{1}+a_{41}\zeta_{1}^{\ast}\alpha_{1}+a_{21}^{\ast}\alpha_{1}^{\ast}\beta_{1}+a_{32}\tau_{1}^{\ast}\beta_{1}+a_{42}\zeta_{1}^{\ast}\beta_{1}+a_{31}^{\ast}\alpha_{1}^{\ast}\tau_{1}+a_{32}^{\ast}\beta_{1}^{\ast}\tau_{1}+a_{43}\zeta_{1}^{\ast}\tau_{1}
+a_{41}^{\ast}\alpha_{1}^{\ast}\zeta_{1}+a_{42}^{\ast}\beta_{1}^{\ast}\zeta_{1}+a_{43}^{\ast}\tau_{1}^{\ast}\zeta_{1})=e^{2\xi_{1}}$, then the two-soliton
solutions in \eqref{lax-93} have the following form
\begin{equation}\label{lax-97}
\begin{split}
q_{1}=&\frac{-2i}{M_{11}M_{22}-M_{12}M_{21}}(-\alpha_{1}e^{\theta_{1}-\theta_{1}^{\ast}}M_{22}+\alpha_{1}e^{\theta_{1}-\theta_{2}^{\ast}}M_{12}\\
&+\alpha_{2}e^{\theta_{2}-\theta_{1}^{\ast}}M_{21}-\alpha_{2}e^{\theta_{2}-\theta_{2}^{\ast}}M_{11}),\\
q_{2}=&\frac{-2i}{M_{11}M_{22}-M_{12}M_{21}}(-\beta_{1}e^{\theta_{1}-\theta_{1}^{\ast}}M_{22}+\beta_{1}e^{\theta_{1}-\theta_{2}^{\ast}}M_{12}\\
&+\beta_{2}e^{\theta_{2}-\theta_{1}^{\ast}}M_{21}-\beta_{2}e^{\theta_{2}-\theta_{2}^{\ast}}M_{11}),\\
q_{3}=&\frac{-2i}{M_{11}M_{22}-M_{12}M_{21}}(-\tau_{1}e^{\theta_{1}-\theta_{1}^{\ast}}M_{22}+\tau_{1}e^{\theta_{1}-\theta_{2}^{\ast}}M_{12}\\
&+\tau_{2}e^{\theta_{2}-\theta_{1}^{\ast}}M_{21}-\tau_{2}e^{\theta_{2}-\theta_{2}^{\ast}}M_{11}),\\
q_{4}=&\frac{-2i}{M_{11}M_{22}-M_{12}M_{21}}(-\zeta_{1}e^{\theta_{1}-\theta_{1}^{\ast}}M_{22}+\zeta_{1}e^{\theta_{1}-\theta_{2}^{\ast}}M_{12}\\
&+\zeta_{2}e^{\theta_{2}-\theta_{1}^{\ast}}M_{21}-\zeta_{2}e^{\theta_{2}-\theta_{2}^{\ast}}M_{11}),
\end{split}
\end{equation}
where
\begin{equation}\label{lax-100}
\left\{
\begin{aligned}
&M_{11}=\frac{-e^{\xi_{1}}}{im_{1}}\cosh(\theta_{1}^{\ast}+\theta_{1}+\xi_{1}),\\
&M_{12}=\frac{-2e^{\xi_{1}}}{n_{2}-n_{1}+i(m_{1}+m_{2})}\cosh(\theta_{1}^{\ast}+\theta_{2}+\xi_{1}),\\
&M_{21}=\frac{-2e^{\xi_{1}}}{n_{1}-n_{2}+i(m_{1}+m_{2})}\cosh(\theta_{2}^{\ast}+\theta_{1}+\xi_{1}),\\
&M_{22}=\frac{-e^{\xi_{1}}}{im_{2}}\cosh(\theta_{2}^{\ast}+\theta_{2}+\xi_{1}).\\
\end{aligned}
\right.
\end{equation}

{\rotatebox{0}{\includegraphics[width=3.6cm,height=3.0cm,angle=0]{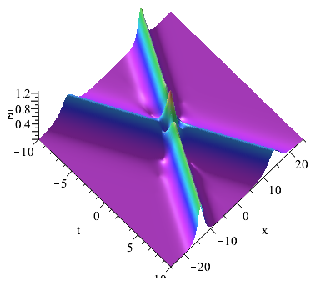}}}
{\rotatebox{0}{\includegraphics[width=3.3cm,height=2.55cm,angle=0]{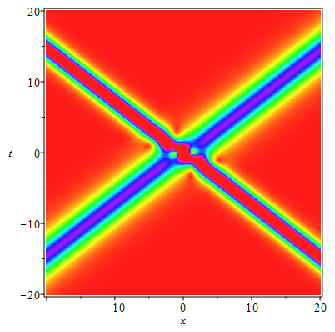}}}
\quad
{\rotatebox{0}{\includegraphics[width=3.6cm,height=2.75cm,angle=0]{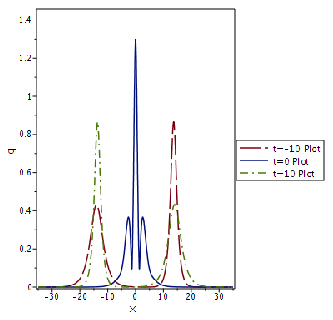}}}

$\quad\qquad\quad(\textbf{a})\quad \ \quad\qquad\qquad\qquad\qquad(\textbf{b})
\quad\qquad\qquad\quad\qquad(\textbf{c})$\\

\noindent { \small \textbf{Figure 5.} Plots of the double-soliton solution $q_{1}$, with the parameters chosen as $a_{11}=a_{22}=a_{33}=a_{44}=0$, $a_{21}=a_{31}=a_{41}=a_{32}=a_{42}=a_{43}=-\frac{1}{9}$, $\alpha_{1}=\tau_{1}=\beta_{1}=\zeta_{1}=\frac{1}{2}-\frac{\sqrt{2}}{2}i$, $\gamma_{1}=\gamma_{2}=1$, $\alpha_{1}=\alpha_{2}$, $\beta_{1}=\beta_{2}$, $\tau_{1}=\tau_{2}$, $\zeta_{1}=\zeta_{2}$ $n_{1}=\frac{-1}{3}$, $n_{2}=\frac{1}{3}$, $m_{1}=0.25$, $m_{2}=0.5$. $\textbf{(a)}$ three dimensional plot at time $t=0$, $\textbf{(b)}$ density plot, $\textbf{(c)}$  the wave propagation along the $x$-axis with different time.}

{\rotatebox{0}{\includegraphics[width=3.6cm,height=3.0cm,angle=0]{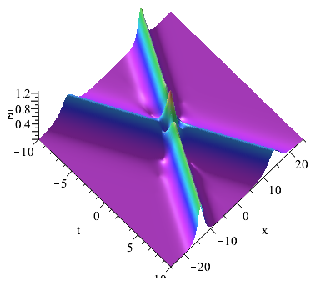}}}
{\rotatebox{0}{\includegraphics[width=3.3cm,height=2.55cm,angle=0]{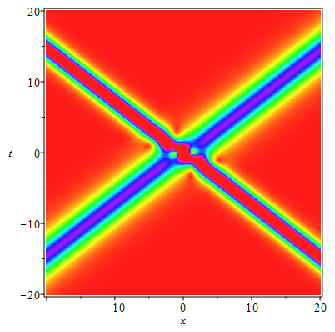}}}
\quad
{\rotatebox{0}{\includegraphics[width=3.6cm,height=2.75cm,angle=0]{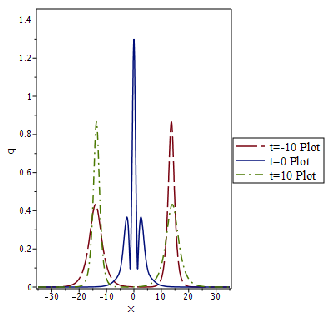}}}

$\quad\qquad\quad(\textbf{a})\quad \ \quad\qquad\qquad\qquad\qquad(\textbf{b})
\quad\qquad\qquad\quad\qquad(\textbf{c})$\\

\noindent { \small \textbf{Figure 6.} Plots of the double-soliton solution $q_{1}$, with the parameters chosen as $a_{11}=a_{22}=a_{33}=a_{44}=0$, $a_{21}=a_{31}=a_{41}=a_{32}=a_{42}=a_{43}=-\frac{1}{9}$, $\alpha_{1}=\tau_{1}=\beta_{1}=\zeta_{1}=\frac{1}{2}-\frac{\sqrt{2}}{2}i$, $\gamma_{1}=\gamma_{2}=1$, $\alpha_{1}=\alpha_{2}$, $\beta_{1}=\beta_{2}$, $\tau_{1}=\tau_{2}$, $\zeta_{1}=\zeta_{2}$ $n_{1}=\frac{-1}{3}$, $n_{2}=\frac{1}{3}$, $m_{1}=0.25$, $m_{2}=0.5$. $\textbf{(a)}$ three dimensional plot at time $t=0$, $\textbf{(b)}$ density plot, $\textbf{(c)}$  the wave propagation along the $x$-axis with different time.}

{\rotatebox{0}{\includegraphics[width=3.6cm,height=3.0cm,angle=0]{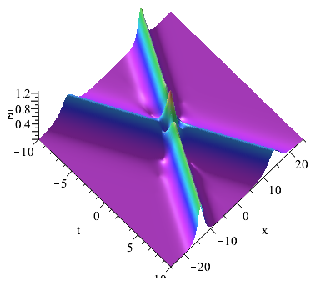}}}
{\rotatebox{0}{\includegraphics[width=3.3cm,height=2.55cm,angle=0]{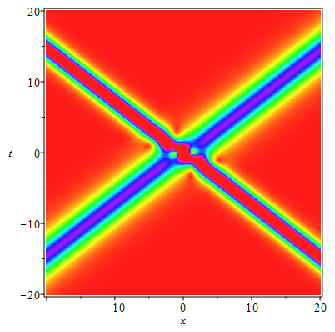}}}
\quad
{\rotatebox{0}{\includegraphics[width=3.6cm,height=2.75cm,angle=0]{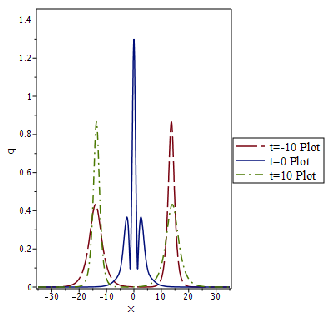}}}

$\quad\qquad\quad(\textbf{a})\quad \ \quad\qquad\qquad\qquad\qquad(\textbf{b})
\quad\qquad\qquad\quad\qquad(\textbf{c})$\\

\noindent { \small \textbf{Figure 7.} Plots of the double-soliton solution $q_{1}$, with the parameters chosen as $a_{11}=a_{22}=a_{33}=a_{44}=0$, $a_{21}=a_{31}=a_{41}=a_{32}=a_{42}=a_{43}=-\frac{1}{9}$, $\alpha_{1}=\tau_{1}=\beta_{1}=\zeta_{1}=\frac{1}{2}-\frac{\sqrt{2}}{2}i$, $\gamma_{1}=\gamma_{2}=1$, $\alpha_{1}=\alpha_{2}$, $\beta_{1}=\beta_{2}$, $\tau_{1}=\tau_{2}$, $\zeta_{1}=\zeta_{2}$ $n_{1}=\frac{-1}{3}$, $n_{2}=\frac{1}{3}$, $m_{1}=0.25$, $m_{2}=0.5$. $\textbf{(a)}$ three dimensional plot at time $t=0$, $\textbf{(b)}$ density plot, $\textbf{(c)}$  the wave propagation along the $x$-axis with different time.}

{\rotatebox{0}{\includegraphics[width=3.6cm,height=3.0cm,angle=0]{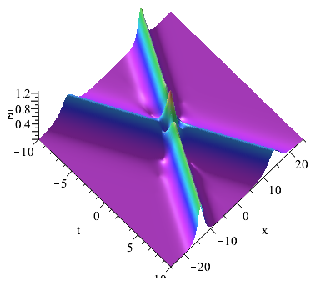}}}
{\rotatebox{0}{\includegraphics[width=3.3cm,height=2.55cm,angle=0]{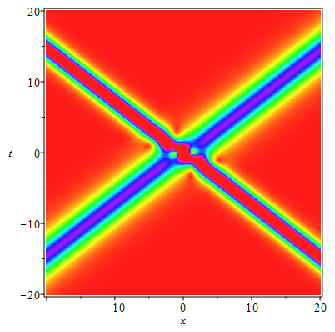}}}
\quad
{\rotatebox{0}{\includegraphics[width=3.6cm,height=2.75cm,angle=0]{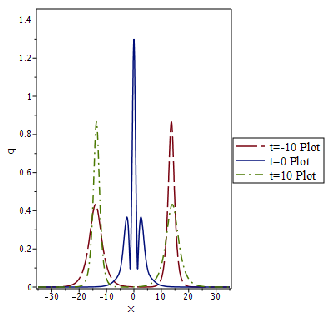}}}

$\quad\qquad\quad(\textbf{a})\quad \ \quad\qquad\qquad\qquad\qquad(\textbf{b})
\quad\qquad\qquad\quad\qquad(\textbf{c})$\\

\noindent { \small \textbf{Figure 8.} Plots of the double-soliton solution $q_{1}$, with the parameters chosen as $a_{11}=a_{22}=a_{33}=a_{44}=0$, $a_{21}=a_{31}=a_{41}=a_{32}=a_{42}=a_{43}=-\frac{1}{9}$, $\alpha_{1}=\tau_{1}=\beta_{1}=\zeta_{1}=\frac{1}{2}-\frac{\sqrt{2}}{2}i$, $\gamma_{1}=\gamma_{2}=1$, $\alpha_{1}=\alpha_{2}$, $\beta_{1}=\beta_{2}$, $\tau_{1}=\tau_{2}$, $\zeta_{1}=\zeta_{2}$ $n_{1}=\frac{-1}{3}$, $n_{2}=\frac{1}{3}$, $m_{1}=0.25$, $m_{2}=0.5$. $\textbf{(a)}$ three dimensional plot at time $t=0$, $\textbf{(b)}$ density plot, $\textbf{(c)}$  the wave propagation along the $x$-axis with different time.}

\section*{Conclusions and discussions}

 In this work, we have proposed a FCNLS equation \eqref{NLS-1} associated with a $5\times5$ Lax pair, which was investigated via the RH approach.
 Based on the Lax pair with a $5\times5$ matrix, we start with the analyze of the spectral problem and the analytical properties of the Jost functions, from which the RH problem of the equation is established. Then, we obtain the $N$-soliton solutions of the FCNLS equation \eqref{NLS-1}, by solving the RH problem without reflection. Finally, we derive two special cases of the solutions to the equation for $N=1$ and $N=2$, and the local structure and dynamic behavior of the one-and two-soliton solutions are analyzed graphically.
\section*{Acknowledgements}

%\textcolor[rgb]{1.00,0.00,0.00}{The authors would like to thank the editor and the referees for their valuable comments and suggestions.}
This work was supported by    the Postgraduate Research and Practice of Educational Reform for Graduate students in CUMT under Grant No. 2019YJSJG046, the Natural Science Foundation of Jiangsu Province under Grant No. BK20181351, the Six Talent Peaks Project in Jiangsu Province under Grant No. JY-059, the Qinglan Project of Jiangsu Province of China, the National Natural Science Foundation of China under Grant No. 11975306, the Fundamental Research Fund for the Central Universities under the Grant Nos. 2019ZDPY07 and 2019QNA35, and the General Financial Grant from the China Postdoctoral Science Foundation under Grant Nos. 2015M570498 and 2017T100413.

\end{document}